\documentclass[a4paper,12pt]{article}

\usepackage[pdftex]{graphicx}
\usepackage{amsmath, amssymb, amsthm, float, mathtools}

\usepackage[T1]{fontenc}
\usepackage[utf8]{inputenc}
\usepackage{authblk}

\usepackage[normalem]{ulem}
\usepackage{soul}
\usepackage{mathrsfs}
\usepackage{macros}
\usepackage{xcolor}
\usepackage{accents}
\usepackage{hyperref}
\usepackage[cm]{fullpage}

\usepackage{enumitem}
\newcommand{\subscript}[2]{$#1 _ #2$}

\allowdisplaybreaks

\AtEndDocument{\bigskip{\footnotesize%
  \textsc{School of Mathematics and Statistics, University of Glasgow, University Place, \\ Glasgow G12 8QQ, UK}
  \par
  \textit{E-mail addresses}: \texttt{Misha.Feigin@glasgow.ac.uk}, \texttt{martinvrabec222@gmail.com}
  \par
  \addvspace{\medskipamount}
}}

\begin{document}

\title{Bispectrality of $AG_2$ Calogero--Moser--Sutherland system}
\author{Misha Feigin, Martin Vrabec}
\maketitle

\begin{abstract} 
We consider the generalised  
Calogero--Moser--Sutherland quantum integrable system associated to the configuration of vectors $AG_2$, which is a union of the root systems $A_2$ and $G_2$. We establish the existence of and construct a suitably defined  Baker--Akhiezer function for the system, and we show that it satisfies bispectrality. We also find two corresponding dual difference operators of rational Macdonald--Ruijsenaars type in an explicit form.
\end{abstract}

\section{Introduction}
\label{intro}
The Calogero--Moser--Sutherland (CMS) systems form an important and rich class of integrable systems with deep relations to geometry, algebra, combinatorics, and other areas. The original works of Calogero, Sutherland and Moser dealt with Hamiltonians describing a pairwise interaction of systems of particles on a circle or a line \cite{C'71, M'75, S'72}. This case was related with the root system $A_n$ by Olshanetsky and Perelomov, who extended CMS Hamiltonians to the case of any root system \cite{OP'76}. In the quantum case, these Hamiltonians are closely related to the radial parts of Laplace--Beltrami operators on symmetric spaces \cite{BPF,OP'78}. 

Ruijsenaars and Schneider introduced a relativistic version of the CMS system in the classical case in \cite{RuijsSchn}. The quantum Hamiltonian and its integrals, which are difference operators, were studied by Ruijsenaars in \cite{Ruijs}. For an arbitrary (reduced) root system, the corresponding commuting operators were introduced by Macdonald \cite{Macdonald2}.  

Ruijsenaars established a duality relation between the classical (trigonometric) CMS system and its (rational)  relativistic version which essentially swaps action and angle variables of the two systems~\cite{Ruijs2}. He also conjectured a duality relation in the quantum case \cite{Ruijs3}. In this case, there exists a function $\psi$ of two sets of variables $x$ and $z$ such that 
\begin{equation}
    \begin{aligned}\label{duality}
        L\psi &= \lambda \psi, \\
        M\psi &= \mu \psi,
\end{aligned}
\end{equation}
where $L = L(x, \partial_x)$ is the CMS Hamiltonian or its quantum integral, $M$ is Ruijsenaars' relativistic CMS operator acting in the variables $z$, and $\lambda = \lambda(z)$, $\mu = \mu(x)$ are the corresponding eigenvalues. The relations~\eqref{duality} may be viewed as a higher-dimensional differential-difference version of the one-dimensional differential-differential bispectrality relation studied by Duistermaat and Gr\"unbaum~\cite{Duijstermaat}.  

The bispectrality relation~\eqref{duality} for the CMS Hamiltonian $L$ related to any root system and for the corresponding Macdonald--Ruijsenaars operators $M$ was established by Chalykh in \cite{Chbisp} in the case of integer coupling parameters. In this case, the function $\psi$ can be taken to be a multi-dimensional Baker--Akhiezer (BA) function \cite{CV, CSV, CFV'99}. It has the form 
\begin{equation}\label{psi}
    \psi (z, x) = P e^{\IP{x}{z}},
\end{equation}
where $P$ is a polynomial in $z$ which also depends on $x$, and $\psi$ can be characterised by its properties as a function of the variables $z$. The relation~\eqref{duality} for the root system $A_n$, another function $\psi$, and an arbitrary coupling parameter was established recently by Kharchev and Khoroshkin in \cite{KK}. We refer to \cite{GVZ} and references therein for aspects of this and various other dualities in the broad realm of integrable systems.

In general, the eigenfunctions of quantum integrable systems may be complicated. However, the multi-dimensional BA functions turn out to be manageable algebraic functions.
An axiomatic definition of BA functions was formulated by Chalykh, Veselov and Styrkas for an arbitrary finite collection of   non-collinear
vectors 
with  integer multiplicities in \cite{CV, CSV} (see also~\cite{Fbisp} for a weaker version of the function; we also refer to \cite{Kr} for one-dimensional BA functions). The key properties are the {\it quasi-invariant} conditions of the form
$$
\psi(z+ s \alpha, x) = \psi(z - s\alpha, x) 
$$
which are satisfied at $\IP{\alpha}{z}=0$ for any vector $\alpha$ from the configuration, and $s$ takes special values.

It was shown that the BA function exists for very special configurations only, which includes positive subsystems of reduced root systems \cite{CSV}. With a slight modification, one can also formulate an axiomatic definition of the BA function in the case of the (only) non-reduced root systems~$BC_n$~\cite{Chbisp}.

The BA function is known to exist for certain generalisations of CMS systems related to special configurations of vectors which are not root systems. 
The known examples are deformations $A_{n,1}(m)$, $C_{n}(m,l)$ of the roots systems $A_n$ and $C_n$, respectively, found by Chalykh, Veselov and one of the authors in \cite{CFV'98, CFV'99}; and another deformation $A_{n,2}(m)$ of the root system $A_{n+1}$ \cite{CVlocus}, which satisfies a weaker axiomatics \cite{Fbisp}. The BA function for the configuration $A_{n,1}(m)$ was constructed in \cite{Chbisp}, and for the cases $C_n(m,l)$, $A_{n,2}(m)$ it was constructed in \cite{Fbisp}. It was also shown in the papers \cite{Chbisp, Fbisp} that the bispectrality relation~\eqref{duality} holds for the corresponding generalisations of Macdonald--Ruijsenaars operators. An important property of these operators is the preservation of a space of quasi-invariant analytic functions.

In the process of classification of algebraically integrable two-dimensional generalised CMS operators, the configuration of vectors $AG_2$ emerged in the work of Fairley and one of the authors \cite{FF}. We studied the corresponding generalised CMS operator  in \cite{FVr}. There we found its eigenfunction $\psi$ of the form \eqref{psi} as $\psi = \mathcal{D} \psi_0$ where $\psi_0$ is the BA function for the root system $G_2$ and $\mathcal{D}$ is an explicit differential operator of order three.

In this paper, we study the function $\psi$ further. We describe it as a BA function by finding the conditions in the variables $z$ which this function satisfies. We find a pair of explicit difference operators which have $\psi$ as their common  eigenfunction. Thus we establish bispectral duality for the generalised CMS system of type $AG_2$. We also obtain an expression for the function $\psi$ by iterations of the action of the difference operators, in analogy with formulas from the works \cite{Chbisp, Fbisp} for other configurations.

Let us now introduce more precisely the generalised CMS operators and the configuration of vectors $AG_2$.
Let $\IP{\cdot}{\cdot}$ denote the standard Euclidean inner product on $\Rn$. We consider the complexification of the real Euclidean space $\Rn$ and extend $\IP{\cdot}{\cdot}$ bilinearly. Let $x = (x_1, \dots, x_n) \in \Cn$. 

Let $A \subset \Cn$ be a finite collection of non-isotropic vectors.
A multiplicity map is a function $m\colon A \to \complex$, $\alpha \mapsto m_\alpha = m(\alpha)$. One calls $m_\alpha$ the multiplicity of $\alpha$. One denotes such collections of vectors with multiplicities by  $\A = (A, m)$, or just $A$ if the intended multiplicity map is clear from the context. 
The corresponding generalised CMS operator has the form
\begin{equation}
\label{genCMS}
    L = -\Delta +  \sum_{\alpha \in A} g_\alpha  \sinh^{-2} \langle \alpha, x \rangle, 
\end{equation}
where $\Delta = \sum_{i=1}^n \partial_{x_i}^2$ is the Laplacian on $\Cn$ and 
we have the convention of writing the constants $g_\alpha$ as 
\begin{equation}\label{eq: convention for coupling}
    g_\alpha = m_\alpha (m_\alpha + 2m_{2\alpha}+ 1) \IP{\alpha}{\alpha}
\end{equation}
with $m_{2\alpha} = 0$ if $2\alpha \notin A$. The convention \eqref{eq: convention for coupling} is used in symmetric spaces theory (see e.g.\ \cite{SV1}).

The Olshanetsky--Perelomov operators correspond to the case when
$A$ is a positive half of a root system and the multiplicity function $m$ is Weyl-invariant. In that case, if 
 $m_\alpha \in \N$ for all $\alpha$ then the operator \eqref{genCMS} admits additional integrals and is algebraically integrable (see \cite{CV, CSV}).

The configuration of vectors $AG_2$ is a non-reduced configuration in $\C^2$ obtained as a union of the root systems $G_2$ and $A_2$. A positive half $AG_{2,+}$ is shown in Figure~\ref{fig: AG_2}, where ${G_{2, +}} = \{\alpha_i, \beta_i\colon i=1, 2, 3\}$, and ${A_{2, +}}=\{2\beta_i\colon i=1, 2, 3\}$. 
There are in total $18$ vectors in the configuration $AG_2$. 
The subscripts on the $\alpha_i$'s are assigned in such a way that $ \langle \alpha_i, \beta_i \rangle = 0$ for all $i = 1, 2, 3$. 
The ratio of the squared lengths of the long roots $\alpha_i$ relative to the short roots $\beta_i$ from the root system $G_2$ is $\a_i^2/\beta_i^2 = 3$ for all $i=1, 2, 3$. 
The vectors $\alpha_i $, $\beta_i$ and $2 \beta_i$ are assigned respectively the multiplicities $m$, $3m$ and $1$, where  $m \in \C$ is a parameter.  

We adopt a coordinate system where the vectors take the form
$$\alpha_1 = \omega\big(0, \sqrt{3} \big), \quad \alpha_2 = \omega\left({-}\tfrac{3}{2}, \tfrac{\sqrt{3}}{2} \right),  \quad \alpha_3 = \omega \left(\tfrac{3}{2}, \tfrac{\sqrt{3}}{2} \right)$$
and
$$\beta_1 = \omega \big(1, 0 \big),\quad \beta_2 = \omega \left(\tfrac{{1}}{2}, \tfrac{\sqrt{3}}{2} \right), \quad \beta_3 = \omega\left({-} \tfrac{{1}}{2}, \tfrac{\sqrt{3}}{2} \right)$$ 
for some scaling $\omega \in \C^\times$.

\begin{figure}
    \centering
    \caption{A positive half of the configuration $AG_2$ \cite{FVr}.}
	\includegraphics[scale=0.6]{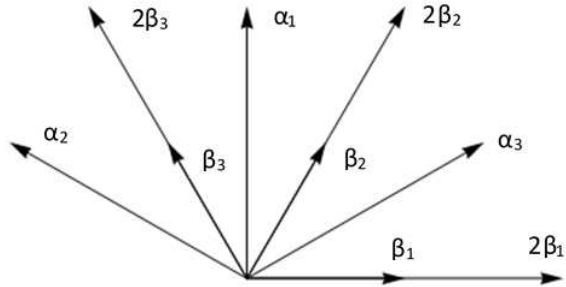}\label{fig: AG_2}
\end{figure}

 We have $\beta_1 + \beta_3 = \beta_2$, $\alpha_2 + \alpha_3 = \alpha_1$, and also
\begin{equation*}
   \begin{aligned}
 &\beta_1 = 2 \beta_2 - \alpha_1 = \alpha_1 - 2 \beta_3 = \alpha_3 - \beta_2 = \beta_3 - \alpha_2, \\
 &\alpha_1 = \tfrac32 \beta_2 + \tfrac12\alpha_2 = \tfrac32\beta_3 + \tfrac12\alpha_3, \text{ and  } \beta_1 = \tfrac12 \beta_2 - \tfrac12 \alpha_2 = -\tfrac12 \beta_3 + \tfrac12 \alpha_3.
\end{aligned}
\end{equation*}

The configuration $AG_2$ is contained in the lattice $\Z\beta_1 \oplus \Z\alpha_2 \subset \C^2$. It is invariant under the Weyl group of type $G_2$, but it is not a crystallographic root system because, for instance, the vectors
$\beta_1$ and $2\beta_2$ have $ 2 \IP{\beta_1}{2\beta_2}/\IP{2\beta_2}{2\beta_2}= \frac12 \notin \integers$.

The corresponding generalised CMS  quantum Hamiltonian \eqref{genCMS} with $A=AG_{2,+}$ has recently been shown to be quantum integrable for any value of the parameter $m \in \C$ \cite{FVr}. It is moreover algebraically integrable for $m \in \N$ by virtue of $AG_2$ being a locus configuration \cite{FF}, as a consequence of the general results presented in \cite{Ch07} (see also~\cite{FVr}). In this paper, we continue investigating the generalised CMS operator for $AG_2$ in the special case when $m \in \N$.

The structure of the paper is as follows. In Section \ref{sec: BA functions}, we recall the notion of a multi-dimensional BA function associated with a configuration $\mathcal A$ following \cite{CV, CSV}, and we generalise it to a case when $A$ has proportional vectors. We show that if such a function exists then it is an eigenfunction for the generalised CMS operator \eqref{genCMS}. 
In Section \ref{sec: bispectral dual difference operators}, we give a general ansatz for the dual generalised Macdonald--Ruijsenaars difference operators with rational coefficients, and we find sufficient conditions for these operators to preserve a space of quasi-invariant analytic functions.
In Section \ref{sec: another difference operator for AG2}, we find a difference operator ${\mathcal D}_1$ related to the configuration $AG_2$ which satisfies the conditions from Section \ref{sec: bispectral dual difference operators}. 
In Section \ref{BAAG2}, we use this operator to show that the BA function for the configuration $AG_2$ exists, and we express this function by iterated action of the operator. We also show that the BA function is an eigenfunction of the operator ${\mathcal D}_1$, thus establishing bispectral duality.
In Section \ref{sec: bispectral dual difference operator for AG2}, we present another dual Macdonald--Ruijsenaars operator ${\mathcal D}_2$ for the configuration $AG_2$, and we establish the corresponding statements for this operator analogous to the ones from Sections \ref{sec: another difference operator for AG2} and~\ref{BAAG2}.

In Section \ref{Sect-A1-A2}, we consider the operator ${\mathcal D}_1$ from Section \ref{sec: another difference operator for AG2} at $m=0$, which gives a Macdonald--Ruijsenaars operator for the root system $A_2$ with multiplicity $1$. We also consider a version of this operator for the root system $A_1$ and decompose it into a sum of two non-symmetric commuting difference operators. We relate these operators with the standard Macdonald--Ruijsenaars operator for the minuscule weight of the root system $A_1$.

\section{Baker--Akhiezer functions}\label{sec: BA functions}

We propose a modified definition of multi-dimensional BA functions so as to extend the definitions from papers \cite{CV, CSV, Chbisp} to include the configuration $AG_2$. We formulate the definition in such a way that it naturally generalises the case of reduced root systems as well as the case of the root system~$BC_n$. 

Let $R \subset \Cn$ be a (possibly non-reduced) finite collection of non-isotropic vectors. We assume there is a subset $R_+\subset R$ such that any collinear vectors in $R_+$ are of the form $\a$, $2\a$ and $R=R_+\sqcup (-R_+)$. Let 
$R^r = \{\a \in R\colon \frac12\a \notin R \}$ and
$R^r_+ = R^r \cap R_+$.
Let $m\colon R\to \Z_{\ge 0}$ be a 
multiplicity map with $m(R^r) \subset \N$.
We extend it to a map $m\colon R \cup 2R^r \to \Z_{\ge 0}$ by putting $m_{2\a} = 0$ if $2\a \notin R$, $\a \in R^r$.

\begin{definition}\label{def: BA modified}
A function $\psi(z, x)$ $(z,x \in \complex^n)$ is a BA function for $(R, m)$ if  
    \begin{enumerate}
        \item $\psi(z,x) = P(z, x) \exp\IP{z}{x}$ for some polynomial $P$ in $z$ with highest order term $\prod_{\gamma \in R_+} \IP{\gamma}{z}^{m_\gamma}$;
        \item $\psi(z + s\alpha, x) = \psi(z-s\alpha, x)$ at $\IP{z}{\alpha} = 0$ for $ s = 1, 2, 3, \dots, m_\alpha, m_\a +2, \dots, m_\a + 2m_{2\a}$ for all $\alpha \in R^r_+$.
    \end{enumerate}
\end{definition}

The multiplicities of the vectors in $AG_2$ are $m_{\beta_i} = 3m$, $m_{2\beta_i} = 1$, $m_{\alpha_i} = m$, and we put $m_{2\alpha_i} = 0$ for all $i=1, 2, 3$, where $m \in \N$. When we apply Definition~\ref{def: BA modified} to this configuration $R = AG_2$, we get the notion of the BA function for $AG_2$.

For $\gamma \in \Cn$, we denote by $\delta_\gamma$ the operator that acts on functions $f(z, x)$ by $$\delta_\gamma f(z,x) = f(z+\gamma,x) - f(z-\gamma,x).$$ 
The condition~2 in Definition~\ref{def: BA modified} admits the following equivalent characterisation.

\begin{lemma}\label{lemma: alt axiomatics abstract} 
     Let $\a \in \Cn$ be non-isotropic, $m_\a \in \N$, and $m_{2\a} \in \integers_{\geq 0}$. 
     A function $\psi(z, x)$ $(z, x \in \Cn)$ analytic in $z$ satisfies
     $\psi(z + s \a,x) = \psi(z - s\a,x) \text{ at }
               \IP{z}{\a} = 0 \textnormal{ for } s = 1, 2, 3, \dots, m_{\a}, m_{\a}+2, \dots, m_{\a} + 2m_{2\a}$
     if and only if
     \begin{equation}\label{eq: alt BA abstract s}
         \bigg( \delta_{\a} \circ \frac{1}{\IP{z}{\a}} \bigg)^{s-1} \delta_{\a} \psi(z, x) = 0 \text{ at } \IP{z}{\a} = 0, s = 1, \dots, m_{\a}
     \end{equation}
     and
     \begin{equation}\label{eq: alt BA abstract t}
        \bigg( \delta_{2\alpha} \circ \frac{1}{\IP{z}{\a}} \bigg)^{t}\circ\bigg( \delta_{\a} \circ \frac{1}{\IP{z}{\a}} \bigg)^{m_\a-1} \delta_{\a} \psi(z, x) = 0 \text{ at } \IP{z}{\a} = 0, t = 1, \dots, m_{2\a}. 
     \end{equation}
\end{lemma}

The proof will follow from the one-dimensional considerations in Lemma~\ref{lemma: alt axiomatics one-dim} below 
 (see also \cite{CSV} where the corresponding statements in the case $m_{2\alpha}=0$ were stated). Let $\delta_r$ ($r \in \N$) be the difference operator that acts on 
 functions $F(k)$ ($k \in \C$) by $\delta_r F(k) = F(k+r) - F(k - r)$. Write $\delta = \delta_1$ for short. The next Lemma will be used in the proof of Lemma \ref{lemma: alt axiomatics one-dim}.

\begin{lemma}\label{lemma: alt axiomatics one-dim - additional lemma}
Suppose $\delta_r F = k \widehat F$ for analytic functions $F(k)$ and $\widehat F(k)$. Suppose also that $F(s-r)=F(r-s)$ for some $s$. Then $\widehat F(s) = \widehat F(-s)$ if and only if  $F (s+r) = F (-s-r)$.
\end{lemma}
\begin{proof}
The statement follows by taking $k=s$ and $k=-s$ in the equality $\delta_r F = k \widehat F$.
\end{proof}

\begin{lemma}\label{lemma: alt axiomatics one-dim}
 The following two properties are equivalent for any analytic function $F(k)$ $(k \in \C)$ and any $n\in \N$, $m \in \integers_{\geq 0}$.
    \begin{enumerate}
        \item For all $s = 1, \dots, n,$
        \begin{equation}\label{eq: alt axiomaticss one-dim s}
            \bigg( \delta \circ \frac1k \bigg)^{s-1} \delta F(k) \bigg\rvert_{k=0} = 0, 
        \end{equation}
        and for all $t = 1, \dots, m,$ 
        \begin{equation}\label{eq: alt axiomaticss one-dim t}
            \bigg( \delta_2 \circ \frac1k \bigg)^{t} \circ \bigg( \delta \circ \frac1k \bigg)^{n-1} \delta F(k) \bigg\rvert_{k=0} = 0.
        \end{equation}
        \item $F(s) = F(-s)$ for all $s = 1, 2, 3, \dots, n, n+2, \dots, n+2m$.
    \end{enumerate}
\end{lemma}
\begin{proof}
Let us firstly prove by induction on $n$ that conditions \eqref{eq: alt axiomaticss one-dim s} are equivalent to the existence of analytic functions $G_1, \ldots, G_n$ such that  $\delta G_q (k) = k G_{q+1}(k)$ for $0\le q \le n-1$, $G_0 = F$, with the conditions 
$G_q(s)=G_q(-s)$ for $1\le s \le n-q$.

For the base case $n=1$, the condition \eqref{eq: alt axiomaticss one-dim s} is equivalent to $F(1)=F(-1)$, which is equivalent to the existence of an analytic function $G_1$ such that $\delta F = k G_1$. Thus the base of induction holds.

Let us now assume that the induction hypothesis is satisfied for some $n\in \N$. The relation \eqref{eq: alt axiomaticss one-dim s} for $s=n+1$ is equivalent to $\delta  G_n (k) =0$ at $k=0$, that is $G_n(1) = G_n(-1)$, which is also equivalent to the existence of an analytic function $G_{n+1}$ such that $\delta G_n = k G_{n+1}$. 

We also have $k G_n = \delta G_{n-1}$ by the induction hypothesis. By applying Lemma \ref{lemma: alt axiomatics one-dim - additional lemma}
with $s=1$ we get $ G_{n-1}(2) = G_{n-1}(-2)$. 
By considering $k G_{n-1} = \delta G_{n-2}$ and using that by induction hypothesis $G_{n-2}(1) = G_{n-2}(-1)$, we obtain by Lemma \ref{lemma: alt axiomatics one-dim - additional lemma} with $s=2$ that $G_{n-2}(3) = G_{n-2}(-3)$. Similarly we obtain $G_q(n-q+1)=G_q(-n+q-1)$ successively for $q=n-3,\ldots, 0$, which completes the proof of the claim by induction.

For any fixed $n$ let us now prove by induction on $m$ that relations  \eqref{eq: alt axiomaticss one-dim s}, 
\eqref{eq: alt axiomaticss one-dim t} are equivalent to the existence of functions $G_1, \ldots, G_n$ as in the first paragraph of the proof with the additional conditions
$G_q(s) = G_q(-s)$ for $s=n-q + 2t$ with $1\le t \le m$, $0\le q \le n-1$,
together with the existence of analytic functions $H_1, \ldots, H_m$ satisfying $\delta_2 H_q = k H_{q+1}$ for $0 \le q \le m-1$, $H_0=G_n$, and $H_q(2s)=H_q(-2s)$ for $1\le s \le m-q$.

The base case $m = 0$ is already done.
Let us now assume that the induction hypothesis is satisfied for some $m \in \Z_{\geq 0}$.
The relation \eqref{eq: alt axiomaticss one-dim t} for $t = m+1$ is equivalent to $\delta_2  H_m(k) =0$ at $k=0$, that is $H_m(2) = H_m(-2)$, which is also equivalent to the existence of an analytic function $H_{m+1}$ such that $\delta_2 H_m = k H_{m+1}$.

We also have $k H_m = \delta_2 H_{m-1}$ by the induction hypothesis. By applying Lemma  \ref{lemma: alt axiomatics one-dim - additional lemma},
we obtain $H_q(2(m-q+1)) = H_q(-2(m-q+1))$ successively for $q=m-1,\ldots, 0$.

By considering $k H_{0} = k G_n =\delta G_{n-1}$ and using that by induction hypothesis $G_{n-1}(2m+1) = G_{n-1}(-2m-1)$, we obtain by Lemma~\ref{lemma: alt axiomatics one-dim - additional lemma} for $s= 2(m+1)$ that $G_{n-1}(2m+3) = G_{n-1}(-2m-3)$. Similarly we obtain $G_q(n-q + 2m + 2) = G_q(-n+q - 2m - 2)$ successively for $q= n-2,\ldots, 0$, which completes the proof by induction on $m$.

It follows that condition $1$ in Lemma~\ref{lemma: alt axiomatics one-dim} implies condition $2$. 

Let us now show that condition $2$ implies condition $1$. Firstly, let us prove it in the case $m=0$ by induction on $n$. The base case $n=1$ is clear. Let us assume it for some $n \in \N$. Suppose now that $F(s) = F(-s)$ for $s = 1, \dots, n, n+1$. By the induction hypothesis and the above analysis in the first part of the proof, we know that there exist functions $G_1, \dots, G_n$ with properties as stated in the first paragraph of the proof.
By using $F(n\pm1) = F(-n\mp1)$ and $\delta F = kG_1$ and applying Lemma~\ref{lemma: alt axiomatics one-dim - additional lemma} with $s = n$, we get $G_1(n) = G_1(-n)$.
Similarly we obtain $G_q(n+1-q) = G_q(-n-1+q)$ successively for $q = 2, \dots, n$. In particular, $G_n(1) = G_n(-1)$, which implies that the relation~\eqref{eq: alt axiomaticss one-dim s} holds for $s = n+1$, completing the induction on $n$. 

Suppose that condition $2$ implies condition $1$ for some $m \in \Z_{\geq 0}$. Now suppose $F$ satisfies condition~$2$ and $F(n+2m+2) = F(-n-2m-2)$. By the induction hypothesis and the above analysis in the first part of the proof, we have existence of functions $G_1, \dots, G_n$ and $H_1, \dots, H_m$ with properties as stated in the fifth paragraph of the proof. By using the assumptions on $F$, the fact that $\delta F = k G_1$ and applying Lemma~\ref{lemma: alt axiomatics one-dim - additional lemma} for $s= n+2m+1$, we get that $G_{1}(n+2m+1) = G_{1}(-n-2m-1)$. Similarly we obtain $G_q(n+2m+2-q) = G_q(-n-2m-2+q)$ successively for $q = 2, \dots, n$, and 
$H_q(2m+2-2q) = H_q(-2m-2+2q)$ successively for $q = 0, \dots, m$. 
In particular, $H_m(2) = H_m(-2)$, which implies that the relation~\eqref{eq: alt axiomaticss one-dim t} holds for $t = m+1$, which completes the proof.
\end{proof}

The following Lemma is a generalisation to the present context of   \cite[Lemma~1]{Fbisp} (see also  \cite[Proposition~1]{CSV}).  
In its proof we use that for a polynomial $p(z)$ we have $p(z+\gamma) = p(z) + \LOT$ by the binomial theorem, where ``$\LOT$'' denotes terms of lower order in $z$ than $\deg p(z)$.
This Lemma will be used below to prove uniqueness of the BA function whenever such a function exists.

\begin{lemma}\label{lemma: h.o.t. of BA modified ax}
Let $\psi(z,x) = P(z, x) \exp\IP{z}{x}$ $(z, x \in \Cn)$,  where $P(z, x)$ is a polynomial in $z$. Suppose that $\psi$ satisfies conditions~\eqref{eq: alt BA abstract s} and~\eqref{eq: alt BA abstract t} for some non-zero $\a \in \Cn$, $m_\a \in \N$, $m_{2\a} \in \integers_{\geq 0}$. Then $\IP{\a}{z}^{m_\a + m_{2\a}}$ divides the highest order term $P_0(z, x)$ of $P(z, x)$.
\end{lemma}
\begin{proof}
    Firstly, the condition \eqref{eq: alt BA abstract s} with $s=1$ gives
    \begin{align*}
        0 = \delta_{\a} \psi(z, x) &= \big(P_0(z, x)(\exp{\IP{\a}{x}} - \exp{\IP{-\a}{x}})  + \LOT \big)\exp{\IP{z}{x}} \\
        &= \big(2 \sinh\IP{\a}{x} P_0(z, x) + \LOT \big)\exp{\IP{z}{x}}
    \end{align*}
    at $\IP{z}{\a} = 0$ for all $x \in \Cn$. This implies that $P_0(z, x)$ and the $\LOT$ must be divisible by $\IP{z}{\a}$. Let $P_0(z, x) = \IP{z}{\a}P_0^{(1)}(z, x)$, and let 
    \begin{equation*}
        \psi^{(1)}(z, x) = \frac{\delta_{\a}\psi(z, x)}{\IP{z}{\a}} = \big(2 \sinh\IP{\a}{x} P_0^{(1)}(z, x) + \LOT \big)\exp{\IP{z}{x}}.
    \end{equation*}
    The condition \eqref{eq: alt BA abstract s} with $s=2$ gives 
    \begin{equation*}
         0 = \delta_{\a} \psi^{(1)}(z, x) = 
         \big(4 \sinh^2\IP{\a}{x} P_0^{(1)}(z, x) + \LOT \big)\exp{\IP{z}{x}}
    \end{equation*}
    at $\IP{z}{\a} = 0$ for all $x \in \C^n$. This implies that $\IP{z}{\a}$ divides $P_0^{(1)}(z, x)$. In particular, $P_0(z, x)$ is divisible by $\IP{z}{\a}^2$:
    $
P_0(z, x)\IP{z}{\a}^{-2} =    P_0^{(2)}(z, x) 
    $
    for some polynomial $P_0^{(2)}$. We let
    $
        \psi^{(2)}(z, x) = {\IP{z}{\a}}^{-1}  \delta_{\a}\psi^{(1)}(z, x).
    $
        By iterating for $s=3, \dots, m_\a$, we get that $P_0(z, x)$ is divisible by $\IP{z}{\a}^{m_{\a}}$, and we recursively obtain polynomials
    $
        P_0^{(s)}(z, x) = P_0^{(s-1)}(z, x)\IP{z}{\a}^{-1} = P_0(z, x)\IP{z}{\a}^{-s}
    $
    and functions
    \begin{equation*}
        \psi^{(s)}(z, x) = \frac{\delta_{\a}\psi^{(s-1)}(z, x)}{\IP{z}{\a}} = \big(2^s \sinh^s\IP{\a}{x} P_0^{(s)}(z, x) + \LOT \big)\exp{\IP{z}{x}}.
    \end{equation*}
    The condition \eqref{eq: alt BA abstract t} with $t=1$ then gives
    \begin{equation*}
         0 = \delta_{2\a} \psi^{m_\a}(z, x) = 
         \big(2^{m_\a + 1} \sinh^{m_\a}\IP{\a}{x}\sinh\IP{2\a}{x} P_0^{(m_\a)}(z, x) + \LOT \big)\exp{\IP{z}{x}}
    \end{equation*}
    at $\IP{z}{\a} = 0$. 
    This implies that $\IP{z}{\a}$ divides $P_0^{(m_\a)}(z, x)$ and the \LOT, in particular, $P_0(z, x)$ is divisible by $\IP{z}{\a}^{m_\a+1}$.
    Continuing this iteratively for $t=2, \dots, m_{2\a}$ completes the proof.
\end{proof}

Lemma~\ref{lemma: h.o.t. of BA modified ax} has the following consequence.
\begin{lemma}\label{lemma: h.o.t. of BA for AG2}
 Let $\psi(z,x) = P(z, x) \exp\IP{z}{x}$ $(z, x \in \Cn)$ satisfy condition 2 in Definition~\ref{def: BA modified}, where $P(z, x)$ is a polynomial in $z$ with highest order term $P_0(z, x)$. Then $\prod_{\gamma \in R_+} \IP{\gamma}{z}^{m_\gamma}$ divides $P_0(z, x)$.
\end{lemma}
\begin{proof}
    Lemma~\ref{lemma: h.o.t. of BA modified ax} gives that $P_0(z, x)$ is divisible by $\IP{z}{\a}^{m_\a + m_{2\a}}$ for all $\a \in R^r_+$. This is a constant multiple of  $\IP{z}{\a}^{m_{\a}} \IP{z}{2\a}^{m_{2\a}}$. 
    The statement follows since we are assuming that collinear vectors in $R_+$ are only of the form $\a$, $2\a$.
\end{proof}
Lemma~\ref{lemma: h.o.t. of BA for AG2} leads to the following uniqueness statement analogous to \cite[Proposition~1]{CSV}  (cf.\ also   \cite[Proposition~1]{Fbisp}) with a completely analogous proof.

\begin{proposition}\label{prop: uniqueness of BA function for AG2}
If a BA function satisfying Definition~\ref{def: BA modified} exists, then it is unique. In particular, if the BA function for $AG_2$ exists then it is unique.
\end{proposition}

The next Theorem generalises Theorem \cite[Theorem~1]{CSV} 
to the present context, and it is proved analogously to how that Theorem is proved. It states that if the BA function satisfying Definition~\ref{def: BA modified} exists, then it is a joint eigenfunction of a commutative ring of differential operators in the variables~$x$. Let us first define an isomorphic ring of polynomials. 

Let $\mathcal{R}$ be the ring of polynomials $p(z)\in \C[z_1, \ldots, z_n]$ satisfying
\begin{align*}
    p(z + s\alpha) = p(z-s\alpha) \text{ at } \IP{z}{\a} = 0, s = 1, 2, 3, \dots, m_\a, m_\a + 2, \dots, m_\a + 2m_{2\a} 
\end{align*}
for all $\a \in R^r_+$ (we remark that this is similar to condition~2 in Definition~\ref{def: BA modified}).
We have $z^2 \in \mathcal{R}$. Indeed, for any $\gamma \in \Cn$, $s \in \naturals$, we have $(z\pm s\gamma)^2 = z^2 \pm 2s \IP{z}{\gamma} + s^2\gamma^2 = z^2 + s^2\gamma^2$ at $\IP{z}{\gamma} = 0$.

For a polynomial $p(z) = p(z_1, \dots, z_n)$, by $p(\partial_x)$ we will mean $p(\partial_{x_1}, \dots, \partial_{x_n})$. For example, if $p(z) = z^2 = z_1^2 + \dots + z_n^2$, then $p(\partial_x) = \Delta_x$ is the Laplace  operator in $n$ dimensions acting in the variables $x$.

The following statement takes places.
\begin{theorem}\label{thm: ring of differential operators for AG2}
If the BA function $\psi(z, x)$ satisfying Definition~\ref{def: BA modified} exists, then for any $p(z) \in \mathcal{R}$ there is a differential operator $L_p( x, \partial_x)$ with highest order term $p(\partial_x)$ such that $$L_p( x, \partial_x) \psi(z, x) = p(z) \psi(z, x).$$ For any $p, q \in \mathcal{R}$, the operators $L_p$ and $L_q$ commute.  
\end{theorem}

The next Lemma will be used below to prove that the differential operator $L_{z^2}$ from Theorem~\ref{thm: ring of differential operators for AG2} corresponding to the element $z^2$ of $\mathcal{R}$ coincides with the generalised CMS Hamiltonian \eqref{genCMS} associated to the configuration $A = R_+$.  
The next Lemma is a generalisation of \cite[Lemma 2]{Fbisp} (see also \cite{CSV}).

\begin{lemma}\label{lemma: n.h.o.t. of BA modified}
Suppose $\psi(z, x) = P(z,x) \exp\IP{z}{x}$ satisfies Definition~\ref{def: BA modified}. 
Let $N = \sum_{\gamma \in R_+} m_\gamma$. Write $P(z,x) = \sum_{i=0}^N P_i(z, x)$ where $P_0(z, x) = \prod_{\gamma \in R_+} \IP{\gamma}{z}^{m_\gamma}$ and $P_i$ are polynomials homogeneous in $z$ with $\deg P_i = N - i$. Then
\begin{equation}\label{eq: n.h.o.t. of BA modified}
    \frac{P_1(z, x)}{P_0(z, x)} = - \sum_{\gamma \in R_+} \frac{m_\gamma(m_\gamma + 2m_{2\gamma} + 1)\gamma^2}{2\IP{\gamma}{z}} \coth \IP{\gamma}{x}.
\end{equation}
\end{lemma}
\begin{proof}
We introduce the following notations
\begin{equation*}
    A(z) = P_0(z, x), \quad A_\alpha(z) = \frac{A(z)}{\IP{\alpha}{z}}, \quad A_{\alpha\beta}(z) = \frac{A(z)}{\IP{\alpha}{z} \IP{\beta}{z}} = A_{\beta \a}(z), \quad \textnormal{etc.} 
\end{equation*}
As a shorthand, we will write $A_{\alpha^2}$ for $A_{\alpha \a}$, etc.
In this notation, the equality \eqref{eq: n.h.o.t. of BA modified} can be written as $P_1 = -\sum_{\gamma \in R_+} \frac{1}{2}m_\gamma(m_\gamma + 2 m_{2\gamma} + 1)\gamma^2 A_\gamma(z) \coth \IP{\gamma}{x}$.

The strategy is as follows. Applying condition~2 in Definition~\ref{def: BA modified} and equating the homogeneous terms of each degree on the right- and left-hand sides, one gets equations for $P_1, \dots, P_N$. We are assuming that $\psi$ satisfies the definition, hence a solution exists.
We are only interested in $P_1$ here.

Fix any $\alpha \in R^r_+$. We will omit arguments of functions whenever convenient, and write $e^\alpha$ for $\exp \IP{\alpha}{x}$.
By the binomial theorem $A(z \pm \alpha)$ has the following homogeneous components
\begin{equation*}\label{Binom theorem for A(z)}
    A(z \pm \alpha) = A(z) \pm \sum_{\beta \in R_+} m_\beta \IP{\alpha}{\beta} A_\beta(z) + \LOT
\end{equation*}
and similarly $P_1(z \pm \alpha, x) = P_1(z, x) + \LOT$

We will use the fact that for this fixed $\a$ condition~2 in Definition~\ref{def: BA modified} admits the equivalent characterisation given in Lemma~
\ref{lemma: alt axiomatics abstract}. We will apply condition~\eqref{eq: alt BA abstract s} for $s = 1, \dots, m_\a$, and then if $m_{2\a} > 0$ also condition~\eqref{eq: alt BA abstract t} for $t = 1, \dots, m_{2\a}$, in order to see what conditions this places on $P_1$.
Applying condition~\eqref{eq: alt BA abstract s} with $s = 1$ gives $0 = \delta_\alpha \psi(z, x) = (P(z + \alpha, x)e^\alpha - P(z - \alpha, x)e^{-\alpha}) \exp\IP{z}{x}$ at $\IP{\alpha}{z} = 0$. This equation can be rearranged (after dividing through by $\exp\IP{z}{x}$) as
\begin{align}
    0 = A(z)(e^\a - e^{-\alpha}) + \sum_{\beta \in R_+} m_\beta \IP{\alpha}{\beta} A_\beta(z)(e^\alpha + e^{-\alpha}) + P_1(z,x)(e^\a - e^{-\alpha}) + \LOT  \label{eq: condition P1 first step expanded}
\end{align}
Note that for $\beta \neq \alpha$ and $\beta \neq 2\alpha$ we have $A_\beta(z) = 0$ at $\IP{\alpha}{z} = 0$. By identifying the homogeneous terms of degree $N-1$ in \eqref{eq: condition P1 first step expanded}, we get at $\IP{\alpha}{z} = 0$ that
\begin{align}
    0 &= \left(m_\alpha \alpha^2 A_\alpha + 2m_{2\alpha} \alpha^2 A_{2\alpha}\right) (e^\alpha + e^{-\alpha}) + P_1(e^\alpha - e^{-\alpha}) \nonumber \\
    &= (m_\alpha  + m_{2\alpha}) \alpha^2 A_\alpha(e^\alpha + e^{-\alpha}) + P_1(e^\alpha - e^{-\alpha}). \label{eq: condition P1 first step}
\end{align}
Notice that this can be rearranged as
$
    P_1 = - (m_\alpha  + m_{2\a})\alpha^2 A_\alpha \coth \IP{\alpha}{x}
$
at $\IP{\a}{z} = 0$.

Assume that $m_\alpha > 1$. Then $A_\alpha(z) = 0 $ at $\IP{\alpha}{z} = 0$, which forces $P_1 = 0$ at $\IP{\alpha}{z} = 0$. This is equivalent to divisibility of $P_1$ by $\IP{\alpha}{z}$. 
We are next going to consider the condition~\eqref{eq: alt BA abstract s} from Lemma~\ref{lemma: alt axiomatics abstract} with $s=2$, and we are now interested in the degree $N-2$ terms in the polynomial part of $(\delta_\alpha \circ \IP{\alpha}{z}^{-1})\delta_\alpha \psi(z, x)$. These terms arise in two ways: firstly they come from the degree $N-2$ terms in the polynomial part of $\IP{\alpha}{z}^{-1} \delta_\alpha \psi$ and further action of $\delta_\a$ on $\exp \IP{\a}{z}$; secondly they come from application of $\delta_\a$ onto the terms in $\IP{\alpha}{z}^{-1} \delta_\alpha \psi$ whose polynomial part has degree $N-1$. Let us examine these two possibilities individually.
From \eqref{eq: condition P1 first step expanded} we see that the terms of degree $N-2$ in $\exp\IP{-z}{x} \IP{\alpha}{z}^{-1} \delta_\alpha \psi$ are
\begin{align}\label{N-2}
    &\sum_{\beta \in R_+} m_\beta \IP{\alpha}{\beta} A_{\alpha\beta}(z)(e^\alpha + e^{-\alpha}) + \frac{P_1(z,x)}{\IP{\alpha}{z}}(e^\a - e^{-\alpha}) \\
    &= (m_\a + m_{2\a})\a^2 A_{\a^2}(z) (e^\alpha + e^{-\alpha}) +  \sum_{\beta \in R_+ \setminus \{ \a, 2\a \}} m_\beta \IP{\alpha}{\beta} A_{\alpha\beta}(z)(e^\alpha + e^{-\alpha}) + \frac{P_1(z,x)}{\IP{\alpha}{z}}(e^\a - e^{-\alpha}) \nonumber.
\end{align}

And from \eqref{eq: condition P1 first step expanded} we see that the term of degree $N-1$ in $\exp\IP{-z}{x}\IP{\alpha}{z}^{-1}\delta_\alpha \psi$ is $A_\alpha(z) (e^\a - e^{-\a})$. Since we need to know how $\delta_\a$ acts on it, we note here that
\begin{align}\label{N-1}
A_\alpha(z \pm \alpha) 
= A_\alpha(z) \pm (m_\alpha + m_{2\alpha} - 1)\alpha^2 A_{\alpha^2}(z)
\pm \sum_{\beta \in R_+ \setminus \{\alpha, 2\alpha \}} m_\beta \IP{\alpha}{\beta}A_{\alpha \beta}(z)  + \LOT
\end{align}

Let us now consider condition~\eqref{eq: alt BA abstract s} from Lemma~\ref{lemma: alt axiomatics abstract} with $s=2$, namely
$ 0 = (\delta_\alpha \circ \IP{\alpha}{z}^{-1})\delta_\alpha \psi(z, x)$ at $\IP{\alpha}{z} = 0$. By making use of formulas~\eqref{N-2} and~\eqref{N-1} this condition can be rearranged as
\begin{equation}\label{eq: condition P1 second step expanded}
    \begin{aligned}
        0 = A_\a(z) (e^\a - e^{-\a})^2 &+ (m_\alpha + m_{2\alpha} - 1)\alpha^2 A_{\alpha^2}(z) (e^\a - e^{-\a})(e^\a + e^{-\a}) 
        \\ &+ (m_\a + m_{2\a})\a^2 A_{\a^2}(z) (e^\alpha + e^{-\alpha})(e^\alpha - e^{-\alpha}) 
        \\&+ \sum_{\beta \in R_+ \setminus \{\alpha, 2\alpha \}} 2m_\beta \IP{\alpha}{\beta}A_{\alpha \beta}(z)(e^\alpha + e^{-\a})(e^\alpha - e^{-\a}) \\
        &+ \frac{P_1(z,x)}{\IP{\a}{z}}(e^\a - e^{-\alpha})^2 + \LOT 
\end{aligned}
\end{equation}
Note that for $\beta \neq \alpha$ and $\beta \neq 2\alpha$ we have $A_{\a\beta}(z) = 0$ at $\IP{\alpha}{z} = 0$.
By identifying the homogeneous terms of degree $N-2$ in \eqref{eq: condition P1 second step expanded} we get at $\IP{\a}{z} = 0$ that 
\begin{align}
    0 = \left( (m_\alpha  + m_{2\alpha}) + (m_\alpha  + m_{2\alpha} -1) \right) \alpha^2 A_{\alpha^2}(z)(e^\alpha - e^{-\alpha})(e^\alpha + e^{-\alpha}) 
    + \frac{P_1}{\IP{\a}{z}}(e^\alpha - e^{-\alpha})^2. \label{eq: condition P1 second step}
\end{align}
If $m_\a > 2$, then $A_{\alpha^2}(z) = 0$ at $\IP{\a}{z} = 0$, so \eqref{eq: condition P1 second step} gives divisibility of $\IP{\a}{z}^{-1}P_1$ by $\IP{\a}{z}$.
Continuing iteratively for $s = 3, \dots, m_\a$ the process above, after the $s = m_\a$ step we get that at $\IP{\a}{z} = 0$ 
\begin{equation}\label{eq: condition P1 mth step}
    \begin{aligned}
    0 = &\left( (m_\alpha  + m_{2\alpha}) + (m_\alpha  + m_{2\alpha} -1)+ \dots +(1 + m_{2\a}) \right) \alpha^2 A_{\alpha^{m_\a}}(z)(e^\alpha - e^{-\alpha})^{m_\a - 1}(e^\alpha + e^{-\alpha}) 
    \\
    &+ \frac{P_1}{\IP{\a}{z}^{m_\a - 1}}(e^\alpha - e^{-\alpha})^{m_\a}. 
\end{aligned}
\end{equation}
We have
\begin{align*}
    (m_\alpha  + m_{2\alpha}) + (m_\alpha  + m_{2\alpha} -1)+ \dots +(1 + m_{2\a}) &= m_\a m_{2\a} + (1 + \dots + m_\a)
    = \frac{m_\a(m_\a +2m_{2\a} + 1)}{2}.
\end{align*}
Therefore \eqref{eq: condition P1 mth step} can be rearranged as 
\begin{equation} \label{eq: condition P1 mth step rearranged}
\frac{P_1}{\IP{\a}{z}^{m_\a - 1}} = -\frac{m_\a(m_\a +2m_{2\a} + 1)}{2} \a^2 A_{\a^{m_\a}}(z) \coth \IP{\a}{x} \textnormal{ at } \IP{\a}{z} = 0.
\end{equation}
If $m_{2\alpha} = 0$ for all $\alpha \in R_+$ for which $\frac{\alpha}{2} \notin R_+$, then for each such $\alpha$ there are no more conditions in Definition~\ref{def: BA modified} to be satisfied, and $P_1$ given by \eqref{eq: n.h.o.t. of BA modified} satisfies \eqref{eq: condition P1 mth step rearranged}. This completes the proof in this case, by uniqueness. Otherwise we continue and consider condition~\eqref{eq: alt BA abstract t}. Note that the highest order term that we are carrying forward is $A_{\a^{m_\a -1}}(z)$.

Assume $m_{2\a} > 0$. Then $A_{\alpha^{m_\a}}(z) = 0$ at $\IP{\a}{z} = 0$, so \eqref{eq: condition P1 mth step rearranged} implies divisibility of $\IP{\a}{z}^{1-m_\a}P_1$ by $\IP{\a}{z}$. 
Consider the condition~\eqref{eq: alt BA abstract t} from Lemma~\ref{lemma: alt axiomatics abstract} with $t=1$. Note that 
\begin{equation*}
    A_{\a^{m_\a }}(z \pm 2\a) = A_{\a^{m_\a }}(z) \pm m_{2\a}(2\alpha)^2 A_{(2\a)\a^{m_\a }} + f(z)
\end{equation*}
where the term $f(z)$ vanishes at $\IP{\a}{z} = 0$.
Thus at this step we obtain
\begin{equation}\label{eq: condition P1 m +1 th step}
    \begin{aligned}
        0 = \,\,&m_{2\a}(2\alpha)^2 A_{(2\a)\a^{m_\a }}(z) (e^{2\alpha} + e^{-2\alpha})(e^\alpha - e^{-\alpha})^{m_\a} 
        \\&+ \frac{m_\a(m_\a +2m_{2\a} + 1)\alpha^2}{2} A_{\alpha^{m_\a + 1}}(z)(e^{2\alpha} - e^{-2\alpha})(e^\alpha - e^{-\alpha})^{m_\a - 1}(e^\alpha + e^{-\alpha}) 
        \\
        &+ \frac{P_1}{\IP{\a}{z}^{m_\a}}(e^{2\alpha} - e^{-2\alpha})(e^\alpha - e^{-\alpha})^{m_\a} 
    \end{aligned}
\end{equation}
at $\IP{\a}{z} = 0$. This can be rearranged as 
\begin{equation}\label{eq: condition P1 m+1 th step rearranged}
    \frac{P_1}{\IP{\a}{z}^{m_\a}} = -m_{2\a}(2\a)^2A_{(2\a)\a^{m_\a }}(z) \coth \IP{2\a}{x} - \frac{m_\a(m_\a +2m_{2\a} + 1)\alpha^2}{2} A_{\alpha^{m_\a + 1}}(z) \coth \IP{\a}{x}.
\end{equation}
Assume $m_{2\a} = 1$, then one can check that $P_1$ given by \eqref{eq: n.h.o.t. of BA modified} satisfies \eqref{eq: condition P1 m+1 th step rearranged}. Otherwise $m_{2\a} > 1$ and we continue recursively.
For $t=2$, for example, to treat the highest order term we need to consider
\begin{equation*}
    A_{\a^{m_\a + 1}}(z \pm 2\a) = A_{\a^{m_\a + 1}}(z) \pm (m_{2\a}-1)(2\alpha)^2 A_{(2\a)\a^{m_\a+1}} + g(z)
\end{equation*}
where the term $g(z)$ vanishes at $\IP{\a}{z} = 0$.

We thus iterate for $t = 2, \dots, m_{2\a}$. One has to use that $1 + \dots + m_{2\a} = \frac{1}{2}m_{2\a}(m_{2\a} + 1)$, and at the end we get
\begin{equation} \label{eq: condition P1 m + m2 step}
    \begin{aligned}
    \frac{P_1}{\IP{\a}{z}^{m_\a + m_{2\a}-1}} = &-\frac{m_{2\a}(m_{2\a} + 1)(2\a)^2}{2} A_{(2\alpha) \a^{m_\a + m_{2\a}-1}} \coth \IP{2\a}{x} \\
    &- \frac{m_\a(m_\a +2m_{2\a} + 1)\alpha^2}{2} A_{\a^{m_\a + m_{2\a}}} \coth \IP{\a}{x} 
\end{aligned}
\end{equation}
at $\IP{\a}{z} = 0$. It is straightforward to check that $P_1$ given by \eqref{eq: n.h.o.t. of BA modified} satisfies \eqref{eq: condition P1 m + m2 step} for all $\a \in R_+$ such that $\frac{\a}{2} \notin R_+$. 

If there existed some other $\widetilde{P}_1$ of degree $N-1$ satisfying~\eqref{eq: condition P1 m + m2 step}, then \eqref{eq: condition P1 m + m2 step} would imply that $P_1 - \widetilde{P}_1$ is divisible by $\IP{\a}{z}^{m_\a + m_{2\a}}$ for all $\a \in R_+$ with $\frac{\a}{2} \notin R_+$. But unless $P_1 - \widetilde{P}_1 = 0$ this would mean that $P_1 - \widetilde{P}_1$ has degree at least $N$ which is not possible. This completes the proof.
\end{proof}

The next Propostion has a completely analogous proof to \cite[Proposition 2]{Fbisp}, it just uses Lemma~\ref{lemma: n.h.o.t. of BA modified} in place of \cite[Lemma 2]{Fbisp} (see also \cite{CSV}).
\begin{proposition}\label{prop: Lz2 relates to generalised H}
With notations and assumptions as in Theorem~\ref{thm: ring of differential operators for AG2},
the element $p(z) = z^2$ of $\mathcal{R}$ corresponds to the differential operator
\begin{equation*}
    L_{z^2} = \Delta - \sum_{\gamma \in R_+} \frac{m_\gamma(m_\gamma + 2m_{2\gamma} + 1)\gamma^2}{\sinh^2 \IP{\gamma}{x}},
\end{equation*}
which coincides (up to sign) with the generalised CMS operator \eqref{genCMS} for~$\A = (R_+, m)$.
\end{proposition}
This then implies integrability of the Hamiltonian and provides an integral of motion $L_p$ for each $p(z) \in \mathcal{R}$. 
The following statement is one of the main results of  this paper. The proof will be presented below.
\begin{theorem}
There exists a BA function for the configuration $AG_2$.
\end{theorem}

\section{Generalised 
Macdonald--Ruijsenaars operators  }\label{sec: bispectral dual difference operators}
In Section \ref{sec: another difference operator for AG2}, we will utilise a method for explicit construction of BA functions which was proposed by Chalykh \cite{Chbisp} (see also \cite{Fbisp} for further examples where this method is applied, and \cite{ChDarboux, CFV'99} for the differential case). The construction uses (generalised)  Macdonald--Ruijsenaars 
difference operators. 
The key element of the method is the preservation of a space of quasi-invariant analytic functions under the action of the Macdonald--Ruijsenaars operators. In this Section, we find sufficient conditions for a sufficiently general difference operator to preserve such a ring of quasi-invariants.

Let $W = \langle s_\a \colon \a \in R \rangle$, where $s_\alpha$ is the orthogonal reflection about the hyperplane $\IP{\alpha}{x}=0$. We assume now that the collection $R$ is $W$-invariant, that is $w(R) = R$ for all $w \in W$, and that the multiplicity map $m$
is $W$-invariant, too. Let $u^\vee = 2 u / \IP{u}{u}$ for any $u\in \Cn$ such that $\IP{u}{u}\ne 0$.

Let $\mathcal{R}^a$ be the ring of analytic functions $p(z)$ such that
\begin{equation*}
    p(z + t\a) = p(z - t\a) \textnormal{ at } \IP{\a}{z} = 0 \textnormal{ for } t \in A_\a
\end{equation*}
for all $\a \in R^r_+$, where $A_\a \subset \naturals$ specifies the axiomatics that one wants to consider. For instance, it can be $A_\a = \{1, 2, 3, \dots, m_\a, m_\a + 2, m_\a + 4, \dots, m_\a + 2m_{2\a} \}$. We assume that $A_{|w\a|} = A_\a$ for all $w \in W$, $\a \in R^r_+$, where $|w\a| = w\a$ if $w\a \in R^r_+$ and $|w\a| = -w\a$ if $w\a \in -R^r_+$. For $\alpha \in R$, we let $\sgn \alpha = 1$ if $\alpha \in R_+$ and $\sgn \alpha =-1$ if $\alpha \in -R_+$.

Let $S \subset \Cn\setminus \{0\}$ be a $W$-invariant finite collection of vectors.
Let $z \in \Cn$ and for any $\gamma \in \Cn$ let $T_\gamma$ be the shift operator that acts on functions $f(z)$ by $T_\gamma f(z) = f(z + \gamma)$.
We are interested in difference operators $D$ of the form 
\begin{equation}\label{eq: generalised RM operator, general form}
    D = \sum_{\tau \in S}  a_\tau(z) (T_\tau - 1), 
\end{equation}
where $a_\tau$ are rational functions with the following three properties:  
\begin{enumerate}[label=(\subscript{D}{{\arabic*}})]
    \item $\deg a_\tau = 0$. \label{enum: degree}
    \item $a_\tau(z)$ has a simple pole at $\IP{\a}{z} = c\a^2$ for $\a \in R^r_+$, $c \in \C$ if and only if $\lambda = s_\a(\tau) - 2c\a  \in S\cup \{0\}$ and 
    $$|c+\IP{\tau}{(2\a)^\vee}| = |\IP{\tau + c\a}{\a}/\a^2|  \in A_\a.$$
    There are no other singularities in $a_\tau$. 
    Denote the set of all such $(\a, c)$ for this $\tau$ by $S_\tau$. 
        \label{enum: singularities}
    \item $wa_\tau = a_{w\tau}$ for all $w \in W$. \label{enum: symmetry}
\end{enumerate}
The condition \ref{enum: singularities} implies that if $a_\tau$ has a singularity $\IP{\alpha}{z}=c \alpha^2$, then for any such $z$ the vectors $z + \tau$ and $z + \lambda$ are of the form $z+\tau = \widetilde z + t\a$ and $z + \lambda = \widetilde z - t\a$ for some $\widetilde z$ with $\IP{\a}{\widetilde z} = 0$ and $|t| = |c + \IP{\tau}{(2\a^\vee)}| \in A_\a$.
We note also that $\lambda \neq \tau$ since $0 \notin A_\a$.

Note that if $a_\tau$ has a singularity $\IP{\alpha}{z}=c \alpha^2$ and the corresponding $\lambda \neq 0$, then by condition~$(D_2)$ $a_\lambda(z)$ necessarily also has a singularity at $\IP{\a}{z} = c\a^2$, since $s_\a(\lambda) - 2c\a = \tau \in S$ and $\IP{\lambda + c\a}{\a} = \IP{s_\a(\tau + c\a)}{\a} = -\IP{\tau + c\a}{\a}$. 
In other words,  by condition $(D_2)$ we have $(\a, c) \in S_\tau$ if and only if $(\a, c) \in S_\lambda$ 
provided that both $\tau, \lambda \neq 0$. 
We additionally observe the following.
\begin{lemma}\label{lemma: symmetry of S_tau}
For any $w \in W$, $(\a, c) \in S_\tau$ if and only if $(|w\a|, \sgn(w\a) c) \in S_{w\tau}$.
\end{lemma}
\begin{proof}
Let $\varepsilon = \sgn(w\a)$. Since $s_{|w\a|} = w s_\a w^{-1}$ and 
$\varepsilon|w\a| = w\a$, we get that
$
    s_{|w\a|}(w\tau) - 2\varepsilon c|w\a| = w(s_\a(\tau) - 2c\a)
$ belongs to $S\cup\{0\}$ if and only if $s_\a(\tau) - 2c\a \in S\cup\{0\}$, by $W$-invariance of $S$. 
Furthermore, $|\IP{w\tau + \varepsilon c|w\a|}{|w\a|}/|w\a|^2| = |\IP{\tau + c\a}{\a}/\a^2|$, and $A_\a = A_{|w\a|}$, by assumption. 
\end{proof}

More explicitly, we are looking at operators of the form 
\begin{equation*}
    D = \sum_{\tau \in S} P_\tau(z) \left( \prod_{(\a, c) \in S_\tau } \left(\IP{\a}{z} - c\a^2 \right)^{-1} \right) (T_\tau-1)
\end{equation*}
for some polynomials $P_\tau(z)$ of degree $|S_\tau|$
so that $\deg a_\tau = 0$ and such that the condition~\ref{enum: symmetry} holds. 
We want to find some sufficient conditions that would ensure that $D$ preserves the ring $\mathcal{R}^a$. 

\begin{theorem}\label{thm: sufficient conditions for ring preservation}
Suppose the operator~\eqref{eq: generalised RM operator, general form} satisfies conditions \ref{enum: singularities} and \ref{enum: symmetry}.
Then for any $\a \in R^r_+$ and arbitrary $p(z) \in \mathcal{R}^a$, we have the following two properties.
\begin{enumerate}
    \item $Dp(z)$ is non-singular at $\IP{\a}{z} = 0$. Moreover, for any $c \neq 0$, provided that for all $\tau \in S$ such that $(\a, c) \in S_\tau$ and such that $\lambda = s_\a(\tau) - 2c\a \neq 0$ we have $$\res_{\IP{\a}{z} = c \a^2}(a_\tau + a_\lambda) = 0,$$ then $Dp(z)$ is non-singular at $\IP{\a}{z} = c\a^2$. 
    \item Suppose, in addition to assumptions of part 1,  that for all $\tau \in S$ and any $t \in A_\a$, the following is satisfied whenever $|t + \IP{\tau}{(2\a)^\vee}| \notin A_\a \cup \{ 0 \}$: 
    \begin{enumerate}
        \item
        $a_\tau(z+t\a) = 0 \textnormal{ at } \IP{\a}{z} = 0$ (equivalently, $P_\tau(z)$ has a factor of $\IP{\a}{z} - t\a^2$),
    \end{enumerate}
    \vspace{-2mm}
    or
    \vspace{-2mm}
    \begin{enumerate}
        \setcounter{enumii}{1}
        \item $\lambda = s_\a(\tau) - 2t\a \in S$ and $a_{\lambda}(z+t\a) = a_\tau(z+t\a)$ at $\IP{\a}{z} = 0$.    
    \end{enumerate}
    
      Then $Dp(z + t\a) = Dp(z-t\a) \textnormal{ at } \IP{\a}{z} = 0 \textnormal{ for all } t \in A_\a.$
\end{enumerate}
\end{theorem}
\begin{proof}
1.
Let $c \in \C$.
We want to show that the residue at $\IP{\a}{z} = c\a^2$ of $Dp(z)$ is zero. Take any $\tau \in S$ such that $(\a, c) \in S_\tau$. Write $\tau + c\a = a\a + \gamma$ where $\IP{\gamma}{\a}=0$ and $a = \IP{\tau + c\a}{\a}/\a^2$. Let $\lambda = s_\a(\tau) - 2c\a$. Then $\lambda + c\a = s_\a(\tau + c\a) = -a\a + \gamma$. At $\IP{\a}{z} = c\a^2$ we thus have  
\begin{align*}
    p(z+\tau) &= p(z-c\a + c\a + \tau) = p((z - c\a + \gamma) + a\a) = p((z - c\a + \gamma) - a\a) = p(z+\lambda) 
\end{align*}
as  $\IP{z-c\a+ \gamma}{\a} = 0$ and $|a| \in A_\a$ by assumption~\ref{enum: singularities}. So, if $\lambda = 0$ then the simple pole at $\IP{\a}{z} = c\a^2$ present in $a_\tau(z)$ is cancelled by $(T_\tau -1)[p(z)] = p(z+\tau)-p(z)$. And if $\lambda \neq 0$, then the sum $$a_\tau(z)(p(z + \tau)-p(z)) + a_\lambda(z)(p(z + \lambda)-p(z))$$ contributes zero to the residue provided that the residue of $a_\tau + a_\lambda$ is zero. For $c \neq 0$ the latter is satisfied by assumption.
In the case of $c = 0$, we have $\lambda = s_\a(\tau)$, hence $a_\tau(s_{\a}(z)) = a_{\lambda}(z)$ by the symmetry \ref{enum: symmetry} of the operator, so we get
\begin{equation*}
\lim_{\IP{\a}{z} \to 0} \IP{\a}{z} a_\tau(z) = \lim_{\IP{\a}{z} \to 0} \IP{\a}{s_{\a}(z)} a_\tau(s_{\a}(z))  
= - \lim_{\IP{\a}{z} \to 0} \IP{\a}{z} a_{\lambda}(z),
\end{equation*}
that is, the residue of $a_\tau(z)$ at $\IP{\a}{z} = 0$ is minus that of $a_{\lambda}(z)$, as needed.

2. Fix $t \in A_\a$. 
By symmetry \ref{enum: symmetry} of the operator, at $\IP{\a}{z} = 0$ we have $a_\mu(z+s\a) = a_{s_\a(\mu)}(z-s\a)$ for all generic $s \in \C$ and generic $z \in \C^2$ with $\IP{\a}{z} = 0$, $\mu \in S$. Using that $s_\a(S) = S$ we can thus write $Dp(z+t\a) - Dp(z-t\a)$ as
\begin{equation}\label{eq: Dp(z+ta) - Dp(z-ta)}
    \lim_{s \to t}\sum_{\mu \in S} a_{\mu}(z+s\a)\bigg(p(z+s\a+\mu) - p(z-s\a+s_\a(\mu)) - p(z + s\a) + p(z-s\a)\bigg).
\end{equation}

Firstly, let us consider any $\tau \in S$ for which $a_{\tau}(z+t\a)$ is non-singular at  $\IP{\a}{z} = 0$ (for generic~$z$). Then the corresponding $\mu = \tau$ term in the sum~\eqref{eq: Dp(z+ta) - Dp(z-ta)} can be simplified to 
\begin{equation}\label{eq: term corresponding to mu=tau}
    a_{\tau}(z+t\a)\left(p(z+t\a+\tau) - p(z-t\a+s_\a(\tau)\right)
\end{equation} 
as $p(z) \in \mathcal{R}^a$.
Let $\tau = b\a + \delta$, where $\IP{\delta}{\a}=0$ and $b = \IP{\tau}{(2\a)^\vee}$. Then $s_\a(\tau) = -b\a + \delta$. Thus 
\begin{equation}\label{eq: p(z+ta+tau) - p(z-ta+sa(tau))}
    p(z+t\a+\tau) - p(z-t\a+s_\a(\tau)) = p(z+\delta + (t+b)\a) - p(z+\delta - (t+b)\a),
\end{equation}
where $\IP{\a}{z+\delta} = 0$. Hence if $|t + b| \in A_\a \cup \{ 0 \}$, then \eqref{eq: p(z+ta+tau) - p(z-ta+sa(tau))} equals zero, and the whole term \eqref{eq: term corresponding to mu=tau} vanishes. Else, we have by assumption two possibilities. If $a_\tau(z+t\a) = 0$ at $\IP{\a}{z} = 0$ then \eqref{eq: term corresponding to mu=tau} vanishes, and if $a_\tau(z+t\a) \ne 0$ then $\lambda = s_\a(\tau)-2t\a \in S\setminus\{\tau\}$ and  
$a_{\lambda}(z+t\a) = a_\tau(z+t\a)$.
        (Note that $|t + \IP{\tau}{(2\a)^\vee}| \notin A_\a \cup \{ 0 \}$ implies that $\lambda \neq 0, \tau$, and due to \ref{enum: singularities} also that $a_{\lambda}(z+t\a)$ and $a_{\tau}(z+t\a)$ are well-defined at $\IP{\a}{z} = 0$ for generic $z$).
In the latter case,
the term corresponding to $\mu = \lambda$ in the sum \eqref{eq: Dp(z+ta) - Dp(z-ta)} can be simplified to
\begin{equation*}
    a_{\lambda}(z+t\a)\left(p(z+t\a+\lambda) - p(z-t\a+s_\a(\lambda)\right) = a_{\tau}(z+t\a) \left( p(z-t\a + s_\a(\tau)) - p(z+t\a + \tau) \right),
\end{equation*} 
which is the negative of \eqref{eq: term corresponding to mu=tau}, hence the terms corresponding to $\mu=\tau$ and $\mu=\lambda$ in \eqref{eq: Dp(z+ta) - Dp(z-ta)} cancel out.

Secondly, let us consider any $\tau \in S$ for which $a_{\tau}(z+t\a)$ is singular at  $\IP{\a}{z} = 0$. Equivalently, $a_\tau(\widetilde{z})$ is singular at $\IP{\a}{\widetilde{z}} = t\a^2$. Hence $(\a, t) \in S_\tau$, and
$|t + b| \in A_\a$ by assumption~\ref{enum: singularities}. It follows that the expression \eqref{eq: p(z+ta+tau) - p(z-ta+sa(tau))} vanishes. 
We can restate this as $$p(z+s\a + \tau) - p(z-s\a + s_\a(\tau)) = (s-t)q(s)$$ for some analytic function $q(s)$ ($s \in \C$).
Similarly, $p(z+t\a) = p(z-t\a)$ can be restated as 
\begin{equation}\label{def of r(s)}
    p(z+s\a) - p(z-s\a) = (s-t)r(s)
\end{equation} 
for some analytic function $r(s)$.
Moreover, we also have
\begin{equation}\label{def of q(2t-s)}
    p(z+s\a + \lambda) - p(z-s\a + s_\a(\lambda)) = p(z-(2t-s)\a + s_\a(\tau)) - p(z + (2t-s)\a + \tau) = (s-t)q(2t-s).
\end{equation}
Suppose firstly that $\lambda \neq 0$. Then in the sum~\eqref{eq: Dp(z+ta) - Dp(z-ta)} the two terms corresponding to $\mu = \tau$ and $\mu = \lambda$ cancel out. Indeed, they equal
\begin{align*}
    &\lim_{s \to t} a_\tau(z+s\a)(s-t)(q(s)-r(s)) + a_\lambda(z+s\a)(s-t)(q(2t-s)-r(s)) \\
    &\quad= \res_{\IP{z}{\a} = t\a^2}(a_\tau + a_\lambda)(q(t)-r(t)) = 0,
\end{align*}
because $\res_{\IP{z}{\a} = t\a^2}(a_\tau + a_\lambda) = 0$ by assumptions of part 1 with $c = t$.
Suppose now that $\lambda = 0$, then $r(s) = q(2t-s)$ by equalities~\eqref{def of r(s)} and \eqref{def of q(2t-s)}. Therefore the term corresponding to $\mu = \tau$ in the sum \eqref{eq: Dp(z+ta) - Dp(z-ta)} equals
$
    \lim_{s \to t} a_\tau(z+s\a)(s-t)\left(q(s)-q(2t-s)\right) = 0.
$

It follows that the sum \eqref{eq: Dp(z+ta) - Dp(z-ta)} vanishes, as required.
\end{proof}
It follows that if the assumptions of both parts 1 and 2 of Theorem~\ref{thm: sufficient conditions for ring preservation} are satisfied for all $\a \in R^r_+$, then $D$ preserves the ring $\mathcal{R}^a$, that is $Dp(z) \in \mathcal{R}^a$ if $p(z) \in \mathcal{R}^a$.

Additionally, we can use the symmetry assumption \ref{enum: symmetry} to reduce the number of conditions that we have to consider in Theorem~\ref{thm: sufficient conditions for ring preservation}. 
The following statements take place. 
\begin{lemma}\label{lemma: symmetry of the residues}
Suppose that condition \ref{enum: symmetry} holds. If $a_\tau + a_\lambda$ has zero residue at $\IP{\a}{z} = c\a^2$, then $a_{w\tau} + a_{w\lambda}$ has zero residue at $\IP{|w\a|}{z} = \sgn(w\a)c |w\a|^2$ (or, equivalently, at $\IP{\a}{w^{-1}z} = c\a^2$) for any $w \in W$.
\end{lemma}
\begin{proof}
By the property \ref{enum: symmetry} we have $a_{w\tau}(z) + a_{w\lambda}(z) = a_{\tau}(w^{-1}z) + a_{\lambda}(w^{-1}z)$, therefore
\begin{align*}
   &\res_{\IP{\a}{w^{-1}z} = c\a^2}(a_{w\tau}(z) + a_{w\lambda}(z) ) = 
    \lim_{\IP{\a}{w^{-1}z} \to c\a^2} (\IP{\a}{w^{-1}z} - c\a^2)(a_{w\tau}(z) + a_{w\lambda}(z)) \\
    &= \lim_{\IP{\a}{w^{-1}z} \to c\a^2} (\IP{\a}{w^{-1}z} - c\a^2)(a_{\tau}(w^{-1}z) + a_{\lambda}(w^{-1}z)) = \res_{\IP{\a}{\widetilde{z}} = c\a^2}(a_{\tau}(\widetilde z) + a_{\lambda}(\widetilde z)) = 0.\qedhere
\end{align*}
\end{proof}
By combining Lemmas~\ref{lemma: symmetry of S_tau}, \ref{lemma: symmetry of the residues} and Theorem \ref{thm: sufficient conditions for ring preservation}, we obtain the following.
\begin{corollary}\label{cor: symmetry of the residues}
    Suppose the operator~\eqref{eq: generalised RM operator, general form} satisfies conditions \ref{enum: singularities} and \ref{enum: symmetry}.
    If the assumptions of part 1 of Theorem~\ref{thm: sufficient conditions for ring preservation} are satisfied for some $\a \in R^r_+$ then $Dp(z)$ is non-singular at $\IP{|w\a|}{z} = \sgn(w\a)c |w\a|^2$ for all $w \in W$.
\end{corollary}
\begin{proof}
By Theorem~\ref{thm: sufficient conditions for ring preservation} part 1, it suffices to check that for any $\widetilde\tau \in S$ and $c\ne 0$ such that $(|w\a|, \sgn(w\a)c) \in S_{\widetilde\tau}$ and such that $\widetilde\lambda = s_{|w\a|}(\widetilde\tau) - 2c w\a \neq 0$ we have that the residue of $a^{}_{\widetilde\tau} + a^{}_{\widetilde\lambda}$ at $\IP{|w\a|}{z} = \sgn(w\a)c |w\a|^2$ is zero. Since $S$ is invariant, we can write $\widetilde\tau = w\tau$ for some $\tau \in S$. Lemma~\ref{lemma: symmetry of S_tau} then gives $(\a, c) \in S_\tau$. Note that $\widetilde{\lambda} = w\lambda$ for $\lambda = s_\a(\tau) - 2c\a$ (in particular, $\lambda \neq 0$ as $\widetilde\lambda \neq 0$).  By assumption, part 1 of Theorem~\ref{thm: sufficient conditions for ring preservation} holds for this $\a$ and this $c$, that is $\res_{\IP{\a}{z} = c\a^2}(a_\tau + a_\lambda) = 0$. Lemma~\ref{lemma: symmetry of the residues} now gives what we need. 
\end{proof}

\begin{lemma}\label{lemma: symmetry in the off-axiomatics}
Suppose the operator~\eqref{eq: generalised RM operator, general form} satisfies conditions \ref{enum: singularities} and \ref{enum: symmetry}.
If the assumptions of parts 1 and 2 of Theorem~\ref{thm: sufficient conditions for ring preservation} are satisfied for some $\a \in R^r_+$, then the assumptions of part 2 are also satisfied for $w\a$ for all $w \in W$ such that $w\a \in R^r_+$.
\end{lemma}
\begin{proof}
Note that $A_{w\a} = A_\a$. Thus we need to prove that whenever for some $t \in A_\a$ and $\widetilde\tau \in S$ we have 
$|t + \IP{\widetilde\tau}{(2w\a)^\vee}| \notin A_{\a} \cup \{0\}$, then either $a_{\widetilde\tau}(z+tw\a) = 0$ at $\IP{w\a}{z} = 0$, or else
$\widetilde\lambda = s_\a(\widetilde{\tau}) - 2tw\a \in S$ and $a_{\widetilde\lambda}(z+tw\a) = a_{\widetilde\tau}(z+tw\a)$ at $\IP{w\a}{z} = 0$. 

So suppose that $|t + \IP{\widetilde\tau}{(2w\a)^\vee}| \notin A_{\a} \cup \{0\}$.
Since $S$ is invariant, we can write $\widetilde{\tau} = w \tau$ for some $\tau \in S$. Note that then
$\widetilde{\lambda} = w(s_\a(\tau) - 2t\a) = w \lambda$.
Note also that $(2w\a)^\vee = w (2\a)^\vee$. Therefore $|t + \IP{\tau}{(2\a)^\vee}| = |t + \IP{\widetilde\tau}{(2w\a)^\vee}| \notin A_{\a} \cup \{0\}$. By assumption, part 2 of Theorem~\ref{thm: sufficient conditions for ring preservation} holds for this $\a$. Suppose first that $a_\tau(\widetilde z+ t\a) = 0$ at $\IP{\a}{\widetilde z} = 0$. By symmetry \ref{enum: symmetry} of the operator, at $\IP{w\a}{z} = 0$ (or, equivalently at $\IP{\a}{w^{-1}z} = 0$) we thus get $a_{\widetilde\tau}(z+tw\a) = a_{\tau}(w^{-1}z + t\a) = 0$, as required. Otherwise, $\lambda \in S$ so
$\widetilde{\lambda} = w \lambda \in S$ by invariance, and at $\IP{w\a}{z} = 0$ property \ref{enum: symmetry} gives $a_{\widetilde\lambda}(z+tw\a) - a_{\widetilde\tau}(z+tw\a) = a_\lambda(w^{-1}z + t\a) - a_\tau(w^{-1}z + t\a) = 0$, as needed.
\end{proof}

\begin{remark}\label{rem: symmetry in the off-axiomatics} Let $\a \in R^r_+$.
Suppose $w \in W$ satisfies $w \a = \a$. Then, for any $\tau \in S$, in part 2 of Theorem~\ref{thm: sufficient conditions for ring preservation} it suffices to check the given conditions for either $\tau$ or $w\tau$, as the other one then follows. Indeed, for any $t \in A_\a$ we have $|t + \IP{w\tau}{(2\a)^\vee}| = |t + \IP{\tau}{(2\a)^\vee}|$. Also $s_\a(w\tau) - 2t\a = w \lambda$, and at $\IP{\a}{z} = 0$ (equivalently, at $\IP{\a}{w^{-1}z} = 0$) by the symmetry \ref{enum: symmetry} we have $a_{w\tau}(z + t\a) = a_{\tau}(w^{-1}z + t\a)$ and in case (b) also $a_{w\lambda}(z + t\a) = a_{\lambda}(w^{-1}z + t\a)$.
\end{remark}

\section{Bispectral dual difference operator for $AG_2$}
\label{sec: another difference operator for AG2}

In this Section, we find a difference operator ${\mathcal D}_1$ satisfying the conditions of Theorem \ref{thm: sufficient conditions for ring preservation} for the configuration $AG_2$.
We define a difference operator acting in the variable $z \in \cplane$ of the form
\begin{equation}\label{eq: another difference operator for AG2}
    {\mathcal D}_1 =\sum_{\tau \colon \frac12 \tau \in G_2} a_\tau(z)(T_\tau -1).
\end{equation}
Let $W$ be the Weyl group of the root system $G_2$. 
For $\tau = 2\varepsilon\beta_j, \varepsilon \in \{\pm 1\}$, we define 
\begin{equation}
    \begin{aligned}\label{eq: another a2beta}
    a_{2\varepsilon\beta_j}(z) &= 
    3 \prod_{\substack{\gamma \in W\beta_1 \\ \IP{ 2\varepsilon\beta_j}{(2\gamma)^\vee} = 1}} \bigg(1-\frac{(3m+2)\gamma^2}{\IP{\gamma}{z}}\bigg)
    \bigg(1+\frac{3m\gamma^2}{\IP{\gamma}{z}+2\gamma^2}\bigg)
    \bigg(1-\frac{(3m-1)\gamma^2}{\IP{\gamma}{z}-\gamma^2}\bigg) \\
    &\times
    \prod_{\substack{\gamma \in W\a_1 \\ \IP{ 2\varepsilon\beta_j}{(2\gamma)^\vee} = 1}} 
    \bigg(1-\frac{m\gamma^2}{\IP{\gamma}{z}}\bigg)
    \times 
    \bigg(1-\frac{(3m+2)\beta_j^2}{\IP{\varepsilon\beta_j}{z}}\bigg)
    \bigg(1-\frac{3m\beta_j^2}{\IP{\varepsilon\beta_j}{z}+\beta_j^2}\bigg).
    \end{aligned}
\end{equation}
For $\tau = 2\varepsilon\a_j$, we define 
\begin{equation}
\begin{aligned}\label{eq: another a2alpha}
    a_{2\varepsilon\alpha_j}(z) &= 
    \prod_{\substack{\gamma \in W\beta_1 \\ \IP{2\varepsilon\alpha_j}{(2\gamma)^\vee} = 3}} 
    \bigg(1-\frac{(3m+2)\gamma^2}{\IP{\gamma}{z}}\bigg)
    \bigg(1-\frac{(3m+1)\gamma^2}{\IP{\gamma}{z}+\gamma^2}\bigg)
    \bigg(1-\frac{3m\gamma^2}{\IP{\gamma}{z}+2\gamma^2}\bigg) \\
    & \times 
    \prod_{\substack{\gamma \in W\a_1 \\ \IP{2\varepsilon\alpha_j}{(2\gamma)^\vee} = 1}} 
    \bigg(1-\frac{m\gamma^2}{\IP{\gamma}{z}}\bigg) \times
    \bigg(1-\frac{m\alpha_j^2}{\IP{\varepsilon\alpha_j}{z}}\bigg)
    \bigg(1-\frac{m\alpha_j^2}{\IP{\varepsilon\alpha_j}{z}+\alpha_j^2}\bigg).
\end{aligned}    
\end{equation}

The following Lemma shows that these functions $a_\tau(z)$ have $G_2$ symmetry.

\begin{lemma}\label{lemma: G2 symmetry of another D}
Let $a_\tau(z)$ be defined as in~\eqref{eq: another a2beta} and~\eqref{eq: another a2alpha}.
Then for all $w \in W$ we have $w a_\tau = a_{w\tau}$.
\end{lemma}
\begin{proof}
For any $w \in W$ we have $w(W\a_1) = W\a_1$, $w(W\beta_1) = W\beta_1$, and $\IP{w\tau}{w\gamma} = \IP{\tau}{\gamma}$ for all $\gamma, \tau \in \C^2$. The multiplicities are invariant, too. The statement follows.
\end{proof}

Define the ring $\mathcal{R}^a_{AG_2}$ of analytic functions $p(z)$ satisfying conditions 
\begin{equation}\label{eq: repeated axiomatics for R for AG2}
    \begin{aligned} 
        &p(z + s\a_j) = p(z - s\a_j) \textnormal{ at } \IP{\a_j}{z} = 0, s=1, 2, \dots, m,  \\
        &p(z + s\beta_j) = p(z - s\beta_j) \textnormal{ at } \IP{\beta_j}{z} = 0, s=1, 2, \dots, 3m, 3m + 2
    \end{aligned}
\end{equation}
for all $j=1,2,3$.

\begin{theorem}\label{thm: another operator for AG2 preserves the ring}
The operator \eqref{eq: another difference operator for AG2} preserves the ring $\mathcal{R}^a_{AG_2}$.
\end{theorem}
\begin{proof}
One can check that this operator has the property ($D_2$) where $S = 2G_2$.
Let $p(z) \in \mathcal{R}^a_{AG_2}$ be arbitrary.
Without loss of generality, we put $\omega = \sqrt{2}$. 
We introduce new coordinates $(A, B)$ on $\C^2$ given by $A = \IP{\a_1}{z}$ and $B = \IP{\beta_1}{z}$.

If $B = 4$ (equivalently, $\IP{\beta_1}{z} = 2\beta_1^2$), then 
$\IP{\beta_2}{z} = 2 + \frac12 A$, $\IP{\beta_3}{z} = -2 + \frac12 A$, $\IP{\a_2}{z} = -6 + \frac12 A$ and $\IP{\a_3}{z} = 6 + \frac12 A$.
The only terms singular at $B = 4$ are $a_{-2\beta_2}$, $a_{-2\a_3}$, $a_{2\beta_3}$ and $a_{2\a_2}$. Note that $s_{\beta_1}(-2\beta_2) -4\beta_1 = -2\a_3$, and we compute that $\res_{B = 4}(a_{-2\beta_2}) = - \res_{B = 4}(a_{-2\a_3})$ equals
\begin{align*}
   -\tfrac{3m(3m+2) (3m+4)  (A-12-12 m)(A + 6m) (A-4+12 m) (A+12m) (A+4+12 m) (A + 12 + 12m)^2}{(A -12) (A-4) A^3 (A + 4) (A + 12)}.
\end{align*}
As $s_{\a_1}(-2\beta_2) = 2\beta_3$ and $s_{\a_1}(-2\a_3) = 2\a_2$, by Lemma~\ref{lemma: symmetry of the residues} with $w =s_{\a_1}$ we get that $a_{2\beta_3} + a_{2\a_2}$ has zero residue at $B = 4$, too.
By Theorem~\ref{thm: sufficient conditions for ring preservation} part 1, there is thus no singularity at $B = 4$ in $ {\mathcal D}_1 p(z)$.

If $B = 2$ (equivalently, $\IP{\beta_1}{z} = \beta_1^2$), then 
$\IP{\beta_2}{z} = 1 + \frac12 A$, $\IP{\beta_3}{z} = -1 + \frac12 A$, $\IP{\a_2}{z} = -3 + \frac12 A$ and $\IP{\a_3}{z} = 3 + \frac12 A$.
The only $\tau \in 2G_2$ for which $a_\tau$ is singular at $B = 2$ and for which the corresponding $\lambda = s_{\beta_1}(\tau) - 2\beta_1 \neq 0$ are $\tau = 2\beta_2, 2\a_2, -2\beta_3, -2\a_3$. 
Note that $s_{\beta_1}(2\beta_2) -2\beta_1 = 2\a_2$, and we compute that $\res_{B = 2} (a_{2\beta_2}) = - \res_{B = 2} (a_{2\a_2})$ equals
\begin{align*}
  \tfrac{6 (m+1) (3m-1) (3m+1) (A-10-12 m) (A-6-12 m) (A-2-12 m) (A+6-12m)^2 (A-6 m) (A+6+12 m)}{(A - 6)(A-2)A(A+2)(A+6)^3}.
\end{align*}
As $s_{\a_1}(2\beta_2) = -2\beta_3$ and $s_{\a_1}(2\a_2) = -2\a_3$, by Lemma~\ref{lemma: symmetry of the residues} we get that $a_{-2\beta_3} + a_{-2\a_3}$ has zero residue at $B = 2$, too.
By Theorem~\ref{thm: sufficient conditions for ring preservation} part 1, there is thus no singularity at $B = 2$ in ${\mathcal D}_1 p(z)$, nor at $B = 0$.

It follows from the above analysis and from the form of the coefficient functions \eqref{eq: another a2beta} and \eqref{eq: another a2alpha} that there are no singularities in ${\mathcal D}_1 p(z)$ at $B = c$ for all $c \geq 0$. By Corollary~\ref{cor: symmetry of the residues}, there is also no singularity in ${\mathcal D}_1 p(z)$ at $\IP{\beta_i}{z} = c$ for all $i=1,2,3$ and all $c \in \C$.

The only singularity at $A = \mathrm{const} > 0$ present in the coefficients $a_\tau$ for some $\tau$ is at $A = 6$ (equivalently, $\IP{\a_1}{z} = \a_1^2$) when  $\tau = -2\a_1$. This singularity cancels in ${\mathcal D}_1 p(z)$ by Theorem~\ref{thm: sufficient conditions for ring preservation} part~1, since the corresponding $\lambda = s_{\a_1}(-2\a_1) - 2\a_1 = 0$.
By Corollary~\ref{cor: symmetry of the residues}, there is also no singularity in ${\mathcal D}_1 p(z)$ at $\IP{\a_i}{z} = c$ for all $i=1,2,3$ and for all $c \in \reals$. This completes the proof that ${\mathcal D}_1 p(z)$ is analytic. 

Let us now show ${\mathcal D}_1 p(z)$ satisfies the axiomatics of $\mathcal{R}^a_{AG_2}$.
We have $A_{\beta_i} = \{1, 2, 3, \dots, 3m, 3m+2\}$ and $A_{\a_i} = \{ 1,2, \dots, m\}$ ($i=1,2,3$). Let us show firstly that 
${\mathcal D}_1 p(z + t\beta_1) = {\mathcal D}_1 p(z-t\beta_1)$ at $\IP{\beta_1}{z} = 0$ for all $t \in A_{\beta_1}$. To do so we will check condition 2 in Theorem~\ref{thm: sufficient conditions for ring preservation} for all $\tau \in 2G_2$. 

Note that $(2\beta_1)^\vee = \frac12 \beta_1$. 
Let $\tau = 2\beta_1$.
Then $|t + \IP{\tau}{(2\beta_1)^\vee}| = t + 2$ which does not belong to $A_{\beta_1} \cup \{0\}$ if and only if $t = 3m-1$ or $t = 3m+2$. But $$a_{2\beta_1}(z+(3m+2)\beta_1) = a_{2\beta_1}(z+(3m-1)\beta_1) = 0 \textnormal{ at } \IP{\beta_1}{z} = 0$$ because $a_{2\beta_1}(z)$ contains the factors $(1-\frac{(3m+2)\beta_1^2}{\IP{\beta_1}{z}})(1-\frac{3m\beta_1^2}{\IP{\beta_1}{z} + \beta_1^2})$. 

Let now $\tau = -2\beta_1$. Then $|t + \IP{\tau}{(2\beta_1)^\vee}| = |t - 2| \in A_{\beta_1} \cup \{0\}$ for all $t\in A_{\beta_1}$, as needed.

Let now $\tau = 2\beta_2$. Then $|t + \IP{\tau}{(2\beta_1)^\vee}| = t + 1$ which does not belong to $A_{\beta_1} \cup \{0\}$ if and only if $t = 3m$ or $t = 3m+2$.
But $$a_{2\beta_2}(z+(3m+2)\beta_1) = a_{2\beta_2}(z+3m\beta_1) = 0 \textnormal{ at } \IP{\beta_1}{z} = 0$$ because $a_{2\beta_2}(z)$ contains the factors $(1-\frac{(3m+2)\beta_1^2}{\IP{\beta_1}{z}})(1-\frac{(3m-1)\beta_1^2}{\IP{\beta_1}{z} - \beta_1^2})$.  

Let now $\tau = -2\beta_2$. Then $|t + \IP{\tau}{(2\beta_1)^\vee}| = t - 1$ which does not belong to $A_{\beta_1} \cup \{0\}$ if and only if $t = 3m+2$. But
$a_{-2\beta_2}(z+(3m+2)\beta_1) = 0$ at $\IP{\beta_1}{z} = 0$ because $a_{-2\beta_2}$ contains the factor $(1+\frac{3m\beta_1^2}{-\IP{\beta_1}{z}+2\beta_1^2})$. 
Since $s_{\a_1}(\beta_1) = \beta_1$, by Remark~\ref{rem: symmetry in the off-axiomatics} there is nothing to check for $\tau = \pm2\beta_3 = s_{\a_1}(\mp2\beta_2)$. 

For $\tau = \pm 2\a_1$, we get $|t + \IP{\tau}{(2\beta_1)^\vee}| = t \in A_{\beta_1}$, as needed.
Similarly for $\tau = 2\a_2$,  $|t + \IP{\tau}{(2\beta_1)^\vee}|$ $= |t - 3| \in A_{\beta_1} \cup \{0\}$ for all $t \in A_{\beta_1}$, as needed.

Finally, let $\tau = -2\a_2$. Then  $|t + \IP{\tau}{(2\beta_1)^\vee}| = t + 3 \notin A_{\beta_1} \cup \{0\}$ if and only if $t = 3m+2$, $t = 3m$ or $t = 3m-2$, but $a_{-2\a_2}(z+t\beta_1) = 0$ at $\IP{\beta_1}{z} = 0$ for those $t$ because $a_{-2\a_2}$ contains the factors 
$(1- \frac{(3m+2)\beta_1^2}{\IP{\beta_1}{z}})(1- \frac{(3m+1)\beta_1^2}{\IP{\beta_1}{z}+\beta_1^2})(1- \frac{3m\beta_1^2}{\IP{\beta_1}{z}+2\beta_1^2})$.
By Remark~\ref{rem: symmetry in the off-axiomatics}, there is nothing to check for $\tau = \pm2\a_3 = s_{\a_1}(\mp2\a_2)$.

Let us show next that 
${\mathcal D}_1 p(z + t\a_1) = {\mathcal D}_1 p(z-t\a_1)$ at $\IP{\a_1}{z} = 0$ for all $t \in A_{\a_1}$. 
By Theorem~\ref{thm: sufficient conditions for ring preservation} it is sufficient to check its condition 2 for all $\tau \in 2G_2$. 

Let $\tau = \pm 2\beta_1$. Then $|t + \IP{\tau}{(2\a_1)^\vee}| = t \in A_{\a_1}$, as needed.

Let now $\tau = 2\beta_2$. Note that $(2\a_1)^\vee = \frac16 \a_1$.  Then $|t + \IP{\tau}{(2\a_1)^\vee}| = t+1 \notin A_{\a_1} \cup \{0\}$ if and only if $t = m$. But $a_{2\beta_2}(z+m\a_1) = 0$ at $\IP{\a_1}{z} = 0$ because $a_{2\beta_2}$ contains the factor $(1-\frac{m\a_1^2}{\IP{\a_1}{z}})$.

Let now $\tau = -2\beta_2$. Then $|t + \IP{\tau}{(2\a_1)^\vee}| = t-1 \in A_{\a_1} \cup \{0\}$ for all $t \in A_{\a_1}$, as needed.
Since $s_{\beta_1}(\a_1) = \a_1$, by Remark~\ref{rem: symmetry in the off-axiomatics} there is nothing to check for $\tau = \pm 2\beta_3 = s_{\beta_1}(\pm 2\beta_2)$.

Let now $\tau = 2\a_1$. Then $|t + \IP{\tau}{(2\a_1)^\vee}| = t+2 \notin A_{\a_1} \cup \{0\}$ if and only if $t = m$ or $t=m-1$. But $a_{2\a_1}(z+t\a_1) = 0$ at $\IP{\a_1}{z} = 0$ for those $t$ because $a_{2\a_1}$ contains the factors $(1-\frac{m\a_1^2}{\IP{\a_1}{z}})(1-\frac{m\a_1^2}{\IP{\a_1}{z} + \a_1^2})$.

Let now $\tau = -2\a_1$. Then $|t + \IP{\tau}{(2\a_1)^\vee}| = |t-2| \in A_{\a_1} \cup \{0\}$ for all $t \in A_{\a_1}$, as needed.

Let now $\tau = 2\a_2$. Then $|t + \IP{\tau}{(2\a_1)^\vee}| = t+1 \notin A_{\a_1} \cup \{0\}$ if and only if $t = m$. But $a_{2\a_2}(z+m\a_1) = 0$ at $\IP{\a_1}{z} = 0$ because $a_{2\a_2}$ contains the factor $(1-\frac{m\a_1^2}{\IP{\a_1}{z}})$.

Finally, for $\tau = -2\a_2$ we get $|t + \IP{\tau}{(2\a_1)^\vee}| = t-1 \in A_{\a_1} \cup \{0\}$ for all $t \in A_{\a_1}$. And by Remark~\ref{rem: symmetry in the off-axiomatics} there is nothing to check for $\tau = \pm 2\a_3 = s_{\beta_1}(\pm 2\a_2)$.

Since all the vectors $\a_i, \beta_i$ are in the $W$-orbit of $\a_1 \cup \beta_1$ the statement follows by Lemma~\ref{lemma: symmetry in the off-axiomatics}.
\end{proof}

Let us now look at the expansion  of the operator~\eqref{eq: another difference operator for AG2}  as $\omega \to 0$. 
It produces the 
rational CMS operator for the root system of type $G_2$ with multiplicity $m$ for the long roots and multiplicity $3m+1$ for the short roots, as the next Proposition shows.
Let $\widetilde{\beta}_j = \omega^{-1}\beta_j$ and $\widetilde{\a}_j = \omega^{-1}\a_j$ $(j=1,2,3)$ with the same multiplicities as $\beta_j$ and $\a_j$, respectively.
\begin{proposition}
We have
\begin{equation*}
   \lim_{\omega \to 0} \frac{{\mathcal D}_1}{72\omega^{2}} =
     \Delta - \sum_{\gamma \in \{ \widetilde{\beta_i}, \widetilde{\a}_i\colon i=1,2,3\}} \frac{2(m_\gamma + m_{2\gamma})}{\IP{\gamma}{z}} \partial_{\gamma},
\end{equation*}
where $\Delta = \partial^2_{z_1}+\partial^2_{z_2}$.
\end{proposition}
\begin{proof}
We have that $T_{\pm 2\beta_j} - 1 = \pm\omega \partial_{2 \widetilde{\beta}_j} + \frac{1}{2}\omega^2 \partial_{2\widetilde{\beta}_j}^2 + \dots$, and similarly for the other shifts. The terms at $\omega$ in the expansion $\omega \to 0$ of the operator ${\mathcal D}_1$ vanish. The terms that are second order in derivatives in the coefficient at $\omega^2$ in the expansion $\omega \to 0$ of the operator ${\mathcal D}_1$ are
\begin{align*}
    &3\sum_{j=1}^3 \partial_{2\widetilde{\beta}_j}^2
    + \sum_{j=1}^3  \partial_{2\widetilde{\a}_j}^2 =  12 \sum_{j=1}^3 \partial_{\widetilde{\beta}_j}^2 + 4\sum_{j=1}^3  \partial_{\widetilde{\a}_j}^2
    = 72 \Delta.
\end{align*}
Let us now consider the terms that are first order in derivatives in the coefficient at $\omega^2$. 
It is easy to see that such terms containing 
$\IP{\widetilde{\beta}_1}{z}^{-1}$
are
\begin{align*}
    &-12(3m+1)\left(2\partial_{2\widetilde{\beta}_1} + \partial_{2\widetilde{\beta}_2} - \partial_{2\widetilde{\beta}_3}
    +\partial_{2\widetilde{\a}_3} - \partial_{2\widetilde{\a}_2}\right) 
    = -144(m_{\beta_1} + m_{2\beta_1}) \partial_{\widetilde{\beta}_1}.
\end{align*}
Altogether, the term at $\omega^2$ in the expansion of the operator ${\mathcal D}_1$ is as required.
\end{proof}

\section{Construction of the \BA for $AG_2$}
\label{BAAG2}

In this Section, we employ the method from \cite{Chbisp} to give a construction of the BA function for the configuration $AG_2$. The BA function will be an eigenfunction for the difference operator from Section~\ref{sec: another difference operator for AG2}, which establishes bispectrality of the CMS $AG_2$ Hamiltonian.

Consider the operator ${\mathcal D}_1$ given by~\eqref{eq: another difference operator for AG2} with functions $a_\tau$ as given in \eqref{eq: another a2beta} and \eqref{eq: another a2alpha}. 
The following Lemma gives a useful way of expanding the functions $a_\tau$.
\begin{lemma}\label{lemma: another a_tau expanded}
Let $a_\tau(z)$ be defined as in~\eqref{eq: another a2beta} and~\eqref{eq: another a2alpha}. Then
\begin{equation}\label{eq: another a_tau expanded}
    a_{\tau}(z) = \kappa_\tau - \kappa_\tau\sum_{\gamma \in G_{2,+}} \frac{\IP{\tau}{\gamma}(m_\gamma + m_{2\gamma})}{\IP{\gamma}{z}} + R_\tau(z)
\end{equation}
where $\kappa_{2\varepsilon\beta_j} = 3$ and $\kappa_{2\varepsilon\a_j} = 1$, and $R_\tau(z)$ is a rational function with $\deg R_\tau \leq -2$.
\end{lemma}
\begin{proof}
For the factors in $a_\tau$ with shifted singularities at $\IP{\gamma}{z} + c = 0$ for $c \neq 0$ we can use that 
\begin{equation*}
    \frac{1}{\IP{\gamma}{z} + c} = \frac{1}{\IP{\gamma}{z}} - \frac{c}{(\IP{\gamma}{z} + c)\IP{\gamma}{z}},  
\end{equation*}
which differs from $\IP{\gamma}{z}^{-1}$ only by a rational function of degree $-2$ which cannot affect the coefficient at $\IP{\gamma}{z}^{-1}$. The relation~\eqref{eq: another a_tau expanded} is then obtained by multiplying out the factors in each of the $a_\tau$.
\end{proof}

The next Lemma is proved by a direct computation that uses Lemma~\ref{lemma: another a_tau expanded}.
We will apply it in the proof of Theorem~\ref{thm: another construction of BA for AG2} below.

\begin{lemma}\label{lemma: one application of another D}
For $\gamma \in G_{2, +}$ let $n_\gamma \in \naturals$ be arbitrary. Let $N = \sum_{\gamma \in G_{2, +} }n_\gamma$. 
Let
\begin{equation*}
    \mu(x) = \sum_{\tau \colon \frac12 \tau \in G_2} \kappa_\tau(\exp\IP{x}{\tau}-1), 
\end{equation*}
where $\kappa_\tau$ are as in Lemma~\ref{lemma: another a_tau expanded}.
Let $A(z) = \prod_{\gamma \in G_{2, +}} \IP{\gamma}{z}^{n_\gamma}$.
Write 
$
    ({\mathcal D}_1 - \mu(x)) [ A(z)\exp{\IP{x}{z}}] = R(x,z)\exp\IP{x}{z}
$
for some rational function $R(x, z)$ in $z$, which has degree less than or equal to $N$. Then
\begin{align*}
    R(x, z) = \sum_{\gamma \in G_{2, +}} \big(n_\gamma-(m_\gamma + m_{2\gamma})\big) 
    \left(\sum_{\tau \colon \frac12 \tau \in G_2} \kappa_\tau \IP{\tau}{\gamma} \exp{\IP{x}{\tau}}\right)A(z)\IP{\gamma}{z}^{-1}  + S(x, z)\end{align*}
    for some rational function $S(x, z)$ in $z$ with degree less than or equal to $N-2$.
\end{lemma}
\begin{proof}
By making use of the expression for $a_\tau(z)$ given in Lemma~\ref{lemma: another a_tau expanded} we get
\begin{align*}
     &{\mathcal D}_1 [ A(z) \exp{\IP{x}{z}}] = 
     \sum_{\tau \colon \frac12 \tau \in G_2} a_\tau(z)(T_\tau -1)
     [A(z) \exp{\IP{x}{z}}] \\ 
     &= \exp{\IP{x}{z}}\sum_{\tau \colon \frac12 \tau \in G_2} a_\tau(z)\bigg( \exp{\IP{x}{\tau}}\prod_{\gamma \in G_{2, +}} (\IP{\gamma}{z} + \IP{\tau}{\gamma})^{n_\gamma} 
     - A(z)   \bigg)
      \\
     &= A(z)\exp{\IP{x}{z}} \sum_{\tau \colon \frac12 \tau \in G_2} \kappa_\tau\left(1 - \sum_{\gamma \in G_{2,+}} \IP{\tau}{\gamma}(m_\gamma + m_{2\gamma})\IP{\gamma}{z}^{-1} + \LOT\right) \\
     &\hspace{14em}\times\left(
     \exp{\IP{x}{\tau}}\bigg(1 + \sum_{\gamma \in G_{2, +}} n_\gamma \IP{\tau}{\gamma}\IP{\gamma}{z}^{-1} + \LOT\bigg) - 1\right) \\
     &= A(z)\exp\IP{x}{z} \left(\mu(x) +
     \sum_{\gamma \in G_{2, +}} \big(n_\gamma-(m_\gamma + m_{2\gamma})\big) 
    \left(\sum_{\tau \colon \frac12 \tau \in G_2} \kappa_\tau \IP{\tau}{\gamma} \exp{\IP{x}{\tau}}\right)\IP{\gamma}{z}^{-1} + \LOT \right),
\end{align*}
where $\LOT$ denotes lower degree terms in $z$, and where we used that $\sum_{\tau \colon \frac12 \tau \in G_2} \kappa_\tau \IP{\tau}{\gamma} = 0$ for all $\gamma \in G_{2, +}$ since if $\frac12\tau \in G_2$ then also $-\frac{1}{2}\tau \in G_2$ and $\kappa_\tau = \kappa_{-\tau}$.
\end{proof}

We are ready to give the main result of this Section.

\begin{theorem}\label{thm: another construction of BA for AG2}
Let $M = \sum_{\gamma \in AG_{2, +}} m_\gamma =  12m+3.$
For $x \in \cplane$, let
\begin{equation*}
    \mu(x) =  \sum_{\tau \colon \frac12 \tau \in G_2} \kappa_\tau(\exp\IP{x}{\tau}-1),
\end{equation*}
and
\begin{equation}\label{eq: another c(x) for AG2}
    c(x) = \frac{M!}{8} \prod_{\gamma \in G_{2, +}} \left(\sum_{\tau \colon \frac12 \tau \in G_2} \kappa_\tau \IP{\tau}{\gamma} \exp{\IP{x}{\tau}} \right)^{m_\gamma + m_{2\gamma}}.
\end{equation}
We also define for $z \in \cplane$ the polynomial belonging to $\mathcal{R}^a_{AG_2}$ given by
\begin{equation*}
    Q(z) = \prod_{\substack{\gamma \in G_{2, +} \\ s \in A_\gamma}}  \left(\IP{\gamma}{z}^2-s^2\gamma^2\right), 
\end{equation*}
where we recall $A_{\gamma} = \{1,2,3, \dots, m_\gamma, m_\gamma + 2m_{2\gamma}\}$ for all $\gamma \in G_{2,+}$.
Then the function
\begin{equation}\label{eq: another formula for BA for AG2}
    \psi(z, x) = c^{-1}(x) ({\mathcal D}_1 - \mu(x))^M [Q(z) \exp\IP{z}{x}]
\end{equation}
is the BA function for $AG_2$.  
Moreover, $\psi$ is also an eigenfunction of the operator  ${\mathcal D}_1$ with ${\mathcal D}_1 \psi = \mu(x) \psi$, thus bispectrality holds.
\end{theorem}

\begin{proof}
The operator ${\mathcal D}_1$ preserves the ring $\mathcal{R}^a_{AG_2}$ by Theorem~\ref{thm: another operator for AG2 preserves the ring}.
The function $Q(z)\exp\IP{z}{x}$ belongs to $\mathcal{R}^a_{AG_2}$ since it is analytic and satisfies conditions~\eqref{eq: repeated axiomatics for R for AG2} given that $Q(z + s\gamma) = Q(z - s\gamma) = 0$ at $\IP{\gamma}{z} = 0$, $s \in A_\gamma$, for all $\gamma \in G_{2, +}$. Since ${\mathcal D}_1$ preserves $\mathcal{R}^a_{AG_2}$, so does ${\mathcal D}_1 - \mu(x)$, hence $\psi(z, x)$ given by \eqref{eq: another formula for BA for AG2} belongs to $\mathcal{R}^a_{AG_2}$. 
Its analyticity and the form of the functions $a_\tau$ imply that it equals $c^{-1}(x)P(z, x) \exp\IP{z}{x}$ for some polynomial $P(z, x)$ in $z$. To prove that $\psi(z, x)$ satisfies the definition of the BA function, it thus suffices to calculate the highest degree term in $P(z, x)$.

The highest degree term in $Q(z)$ is $Q_0(z) = \prod_{\gamma \in G_{2, +}} \IP{\gamma}{z}^{2(m_\gamma + m_{2\gamma})}$ and $\deg Q_0 = 2M$. 
For all $k \in \N$ with $k \leq M$, an analogous argument as above gives that $({\mathcal D}_1 - \mu(x))^k[Q(z)\exp\IP{z}{x}]$ belongs to $\mathcal{R}^a_{AG_2}$ and is of the form $Q^{(k)}(z, x) \exp\IP{z}{x}$ for some polynomial $Q^{(k)}(z, x)$ in $z$. Let its highest-degree homogeneous component be $Q_0^{(k)}(z, x)$.
Lemma~\ref{lemma: one application of another D} allows to compute $Q_0^{(k)}(z, x)$. 

Lemma~\ref{lemma: one application of another D} gives that after the first application of ${\mathcal D}_1-\mu(x)$ onto $Q(z)\exp\IP{z}{x}$ we get
$$
Q_0^{(1)} = \sum_{\gamma \in G_{2, +}}(m_\gamma + m_{2\gamma})\left(\sum_{\tau \colon \frac12 \tau \in G_2} \kappa_\tau \IP{\tau}{\gamma} \exp{\IP{x}{\tau}}\right)\IP{\gamma}{z}^{-1}Q_0(z).
$$ 
The second application gives
\begin{align*}
    &Q_0^{(2)} = \sum_{\gamma \in G_{2, +}}(m_\gamma + m_{2\gamma})(m_\gamma + m_{2\gamma}-1)\left(\sum_{\tau \colon \frac12 \tau \in G_2} \kappa_\tau \IP{\tau}{\gamma} \exp{\IP{x}{\tau}}\right)^2\IP{\gamma}{z}^{-2}Q_0(z) \\
    &+\sum_{\substack{\gamma \neq \delta \\ \gamma, \delta \in G_{2, +}}}(m_\gamma + m_{2\gamma})(m_\delta + m_{2\delta})
    \left(\sum_{\tau \colon \frac12 \tau \in G_2} \kappa_\tau \IP{\tau}{\gamma} \exp{\IP{x}{\tau}}\right)
    \\
    &\hspace{5em}\times \left(\sum_{\tau \colon \frac12 \tau \in G_2} \kappa_\tau \IP{\tau}{\delta} \exp{\IP{x}{\delta}}\right)\IP{\gamma}{z}^{-1}\IP{\delta}{z}^{-1}Q_0(z).
\end{align*}
By repeatedly applying Lemma~\ref{lemma: one application of another D} we get 
$$
Q_0^{(k)} = \sum_{\mathbf{n}} f_{\mathbf{n}}(x) Q_0(z) \prod_{\gamma \in G_{2, +}}\IP{\gamma}{z}^{-n_\gamma}
$$ 
where $\mathbf{n} = (n_\gamma)_{\gamma \in G_{2,+}}$ for $n_\gamma \in \integers_{\geq0}$ such that $n_\gamma$ add up to $k$ and where $f_{\bf{n}}(x)$ is non-zero only if $n_\gamma \leq m_\gamma + m_{2\gamma}$ for all $\gamma$.
It follows that $\deg P \leq M$ and that the highest degree term of $P(z, x)$ is $$d(x) \prod_{\gamma \in G_{2, +}} \IP{\gamma}{z}^{m_\gamma + m_{2\gamma}} = \frac{1}{8} d(x)\prod_{\gamma \in AG_{2,+}} \IP{\gamma}{z}^{m_\gamma}$$ for some function $d(x)$. It also implies that the polynomial part of $({\mathcal D}_1 - \mu(x))^{M + 1}[Q(z)\exp\IP{z}{x}]$ has degree less than $M$ hence vanishes as a consequence of Lemma~\ref{lemma: h.o.t. of BA for AG2}, giving ${\mathcal D}_1 \psi = \mu(x)\psi$. So to complete the proof we just need to verify that $c(x)$ given by~\eqref{eq: another c(x) for AG2} equals $\frac{1}{8}d(x)$. 

To arrive at $\prod_{\gamma \in G_{2, +}} \IP{\gamma}{z}^{m_\gamma + m_{2\gamma}}$ starting from $Q_0(z)$ we overall need to reduce the power on each of the factors $\IP{\gamma}{z}$ by $m_\gamma + m_{2\gamma}$ and we do this by reducing the power of one of them by one at each step. 
The total number of possible orderings of doing that corresponds to the number of words of length $M$ in the alphabet $G_{2, +}$ such that $\gamma$ appears in the word $m_\gamma + m_{2\gamma}$ times for each $\gamma \in G_{2, +}$.
This gives
$$\frac{M!}{\prod_{\gamma \in G_{2, +}}(m_\gamma + m_{2\gamma})!}$$ possibilities, and for each of them the total proportionality factor picked up equals by Lemma~\ref{lemma: one application of another D}
$$
\prod_{\gamma \in G_{2, +}}(m_\gamma + m_{2\gamma})!
\left(\sum_{\tau \colon \frac12 \tau \in G_2} \kappa_\tau \IP{\tau}{\gamma} \exp{\IP{x}{\tau}} \right)^{m_\gamma + m_{2\gamma}}
\hspace{-3em}.
$$
It follows that $c(x)$ has the required form.
\end{proof}

\section{Another dual operator
}
\label{sec: bispectral dual difference operator for AG2}

In this Section, we present another difference operator for the configuration $AG_2$ which preserves the quasi-invariants. We also give the corresponding second construction of the BA function.

We define a difference operator acting in the variable $z \in \cplane$ of the form
\begin{equation}\label{eq: difference operator for AG2}
    {\mathcal D}_2 =\sum_{\tau \colon \frac12 \tau \in AG_2} a_\tau(z)(T_\tau -1).
\end{equation}
We now proceed to specify the functions $a_\tau(z)$. Let $\lambda_\tau = g_{\tau/2}$ where $g$ is defined in terms of the multiplicity map of $AG_2$ in accordance with the convention~\eqref{eq: convention for coupling} for couplings.
That is, we set $$ \lambda_\tau = \frac14 m_{\frac12 \tau}(m_{\frac12 \tau} + 2m_{\tau} +1)\tau^2.$$
Recall that the multiplicities are $m_{\varepsilon\beta_j} = 3m$, $m_{2\varepsilon\beta_j} = 1$ and $m_{\varepsilon\alpha_j} = m$ ($j=1,2,3$, $\varepsilon \in \{ \pm 1\}$, $m \in \N$), and we put $m_{2\varepsilon\alpha_j} = 0$. 
That means $\lambda_{4\varepsilon \beta_j} = 8\beta_j^2$, $\lambda_{2\varepsilon\beta_j} = 9m(m+1) \beta_j^2$ 
and $\lambda_{2\varepsilon\alpha_j}  =  m(m+1)\alpha_j^2$.
For $\tau = 4\varepsilon\beta_j$ we define
\begin{equation}\label{eq: a4beta}
\begin{aligned}
     &a_{4\varepsilon\beta_j}(z) = 
    \lambda_{4\varepsilon\beta_j}
    \prod_{\substack{\gamma \in W\a_1 \\ \IP{ 4\varepsilon\beta_j}{(2\gamma)^\vee} = 2}} 
    \bigg(1-\frac{m\gamma^2}{\IP{\gamma}{z}}\bigg)
    \bigg(1-\frac{m\gamma^2}{\IP{\gamma}{z}+\gamma^2}\bigg)
    \\
    &\times \prod_{\substack{\gamma \in W\beta_1 \\ \IP{ 4\varepsilon\beta_j}{(2\gamma)^\vee} = 2}} \bigg(1-\frac{(3m+2)\gamma^2}{\IP{\gamma}{z}}\bigg)
    \bigg(1-\frac{3m\gamma^2}{\IP{\gamma}{z}+\gamma^2}\bigg) \\
    &\times 
    \bigg(1-\frac{(3m+2)\beta_j^2}{\IP{\varepsilon\beta_j}{z}}\bigg)
    \bigg(1-\frac{3m\beta_j^2}{\IP{\varepsilon\beta_j}{z}+\beta_j^2}\bigg)
    \bigg(1-\frac{(3m+2)\beta_j^2}{\IP{\varepsilon\beta_j}{z}+2\beta_j^2}\bigg)
    \bigg(1-\frac{3m\beta_j^2}{\IP{\varepsilon\beta_j}{z}+3\beta_j^2}\bigg).
\end{aligned}
\end{equation}
For $\tau = 2\varepsilon\beta_j$ we define 
\begin{equation}\label{eq: a2beta}
\begin{aligned}
    &a_{2\varepsilon\beta_j}(z) = 
    \lambda_{2\varepsilon\beta_j}
    \prod_{\substack{\gamma \in W\a_1 \\ \IP{ 2\varepsilon\beta_j}{(2\gamma)^\vee} = 0}} 
    \bigg(1-\frac{\frac23 \gamma^2}{\IP{\gamma}{z}-\gamma^2}\bigg)
    \prod_{\substack{\gamma \in W\a_1 \\ \IP{ 2\varepsilon\beta_j}{(2\gamma)^\vee} = 1}} 
    \bigg(1-\frac{m\gamma^2}{\IP{\gamma}{z}}\bigg) \\
    &\times\prod_{\substack{\gamma \in W\beta_1 \\ \IP{ 2\varepsilon\beta_j}{(2\gamma)^\vee} = 1}} \bigg(1-\frac{(3m+2)\gamma^2}{\IP{\gamma}{z}}\bigg)
    \bigg(1+\frac{3m\gamma^2}{\IP{\gamma}{z}+2\gamma^2}\bigg)
    \bigg(1-\frac{(3m-1)\gamma^2}{\IP{\gamma}{z}-\gamma^2}\bigg) \\
    &\times 
    \bigg(1-\frac{(3m+2)\beta_j^2}{\IP{\varepsilon\beta_j}{z}}\bigg)
    \bigg(1-\frac{3m\beta_j^2}{\IP{\varepsilon\beta_j}{z}+\beta_j^2}\bigg)
    \bigg(1+\frac{4\beta_j^2}{\IP{\varepsilon\beta_j}{z}+3\beta_j^2}\bigg)
    \bigg(1-\frac{4\beta_j^2}{\IP{\varepsilon \beta_j}{z}-\beta_j^2}\bigg).
\end{aligned}
\end{equation}
For $\tau = 2\varepsilon\a_j$ we define 
\begin{equation}\label{eq: a2alpha}
\begin{aligned}
    a_{2\varepsilon\alpha_j}(z) &= 
    \lambda_{2\varepsilon\alpha_j}
    \prod_{\substack{\gamma \in W\beta_1 \\ \IP{2\varepsilon\alpha_j}{(2\gamma)^\vee} = 3}} 
    \bigg(1-\frac{(3m+2)\gamma^2}{\IP{\gamma}{z}}\bigg)
    \bigg(1-\frac{(3m+1)\gamma^2}{\IP{\gamma}{z}+\gamma^2}\bigg)
    \bigg(1-\frac{3m\gamma^2}{\IP{\gamma}{z}+2\gamma^2}\bigg) \\
    & \times 
    \prod_{\substack{\gamma \in W\beta_1 \\ \IP{2\varepsilon\alpha_j}{(2\gamma)^\vee} = 0}} 
    \bigg(1-\frac{6\gamma^2}{\IP{\gamma}{z}-\gamma^2}\bigg) 
    \prod_{\substack{\gamma \in W\a_1 \\ \IP{2\varepsilon\alpha_j}{(2\gamma)^\vee} = 1}} 
    \bigg(1-\frac{m\gamma^2}{\IP{\gamma}{z}}\bigg) \\
    &\times
    \bigg(1-\frac{m\alpha_j^2}{\IP{\varepsilon\alpha_j}{z}}\bigg)
    \bigg(1-\frac{m\alpha_j^2}{\IP{\varepsilon\alpha_j}{z}+\alpha_j^2}\bigg).
\end{aligned}
\end{equation}

Lemma~\ref{lemma: G2 symmetry of D} below shows that the functions $a_\tau(z)$ have $G_2$ symmetry.

\begin{lemma}\label{lemma: G2 symmetry of D}
Let $a_\tau(z)$ be defined as above.
Then for all $w \in W$ we have $w a_\tau = a_{w\tau}$.
\end{lemma}
\begin{proof}
For any $w \in W$ we have $w(W\a_1) = W\a_1$, $w(W\beta_1) = W\beta_1$, and $\IP{w\tau}{w\gamma} = \IP{\tau}{\gamma}$ for all $\gamma, \tau \in \C^2$. The multiplicities are invariant, too, and $\lambda_{w\tau} = \lambda_\tau$ for any $\tau$ such that $\frac12 \tau \in AG_2$. 
The statement follows.
\end{proof}

\begin{theorem}\label{secondDpreserves}
The operator \eqref{eq: difference operator for AG2} preserves the ring $\mathcal{R}^a_{AG_2}$.
\end{theorem}
\begin{proof}
 One can check that the operator satisfies condition \ref{enum: singularities}.
Let $p(z) \in \mathcal{R}^a_{AG_2}$ be arbitrary.
Without loss of generality, we put $\omega = \sqrt{2}$. We introduce new coordinates $(A, B)$ on $\C^2$ given by $A = \IP{\a_1}{z}$ and $B = \IP{\beta_1}{z}$.

It follows from the form of the coefficient functions \eqref{eq: a4beta}, \eqref{eq: a2beta} and \eqref{eq: a2alpha} and Theorem \ref{thm: sufficient conditions for ring preservation}
that there are no singularities in ${\mathcal D}_2p(z)$ at $B = c$ for $c \geq 0$ unless $B=2, 4, 6$. Let us consider these cases.

If $B = 6$ (equivalently, $\IP{\beta_1}{z} = 3\beta_1^2$), then 
$\IP{\beta_2}{z} = 3 + \frac12 A$, $\IP{\beta_3}{z} = -3 + \frac12 A$, $\IP{\a_2}{z} = -9 + \frac12 A$ and $\IP{\a_3}{z} = 9 + \frac12 A$. 
The only terms singular at $B=6$ are $a_{-4\beta_1}$ and $a_{-2\beta_1}$. We note that $s_{\beta_1}(-4\beta_1) - 6\beta_1 = -2\beta_1$, and we compute that $\res_{B = 6}(a_{-4\beta_1}) = - \res_{B = 6}(a_{-2\beta_1})$ equals
\begin{align*}
    \tfrac{48 m (m+1) (3m+2) (3m+5) (A - 2- 12m) (A - 6- 12m) (A - 14- 12m) (A - 18- 12m) (A+2+12 m) (A+6+12 m) (A + 14+12 m) (A+18+12 m)}{(A - 18)(A - 6)^2(A - 2)(A +2)(A +6)^2(A + 18)}.
\end{align*}
Therefore, by Theorem~\ref{thm: sufficient conditions for ring preservation} part 1, there is no singularity at $B = 6$ in ${\mathcal D}_2p(z)$.

If $B = 4$ (equivalently, $\IP{\beta_1}{z} = 2\beta_1^2$), then 
$\IP{\beta_2}{z} = 2 + \frac12 A$, $\IP{\beta_3}{z} = -2 + \frac12 A$, $\IP{\a_2}{z} = -6 + \frac12 A$ and $\IP{\a_3}{z} = 6 + \frac12 A$.
The only $\tau \in AG_2$ for which $a_\tau$ is singular at $B = 4$ and for which the corresponding $\lambda = s_{\beta_1}(\tau) - 4\beta_1 \neq 0$
are $\tau = -2\beta_2, -2\a_3, 2\beta_3$ and $2\a_2$. We note that $s_{\beta_1}(-2\beta_2) - 4\beta_1 = -2\a_3$, and we compute  that 
$\res_{B = 4}(a_{-2\beta_2}) = - \res_{B = 4}(a_{-2\a_3})$ equals
\begin{align*}
   -\tfrac{18m^2(m+1) (3m+2) (3m+4)  (A-32) (A + 24) (A-12-12 m)(A + 6m) (A-4+12 m) (A+12m) (A+4+12 m) (A + 12 + 12m)^2}{(A -12) (A-8) (A-4) A^4 (A + 4) (A + 12)}.
\end{align*}
Since $s_{\a_1}(-2\beta_2) = 2\beta_3$ and $s_{\a_1}(-2\a_3) = 2\a_2$, by Lemma~\ref{lemma: symmetry of the residues} the residue of 
$a_{2\beta_3} + a_{2\a_2}$ at $B = 4$ is also zero.
Thus, by Theorem~\ref{thm: sufficient conditions for ring preservation} part 1, there is no singularity at $B = 4$ in ${\mathcal D}_2p(z)$.

If $B = 2$ (equivalently, $\IP{\beta_1}{z} = \beta_1^2$), then 
$\IP{\beta_2}{z} = 1 + \frac12 A$, $\IP{\beta_3}{z} = -1 + \frac12 A$, $\IP{\a_2}{z} = -3 + \frac12 A$ and $\IP{\a_3}{z} = 3 + \frac12 A$.
The only $\tau \in AG_2$ for which $a_\tau$ is singular at $B = 2$ and for which the corresponding $\lambda = s_{\beta_1}(\tau) - 2\beta_1 \neq 0$
are $\tau = -4\beta_1, 2\beta_1, -4\beta_2, 4\beta_3, \pm 2\a_1, 2\beta_2, 2\a_2, -2\beta_3$ and $-2\a_3$.
We note that $s_{\beta_1}(-4\beta_1) - 2\beta_1 = 2\beta_1$, and we compute that 
$\res_{B = 2}(a_{-4\beta_1}) = - \res_{B = 2} (a_{2\beta_1})$ equals
\begin{align*}
  \tfrac{144 m(m+1) (3m-2) (3m+1) (A +6-12m) (A +2-12m) (A -6-12m) (A-10-12m) (A-6+12 m) (A-2+12 m) (A+6+12 m) (A+10+12 m)}{(A-6)^2(A-2)^2(A+2)^2(A+6)^2}.
\end{align*}
Similarly, we note that $s_{\beta_1}(-4\beta_2) - 2\beta_1 = -2\a_1$, and we compute that 
$\res_{B = 2}(a_{-4\beta_2}) = - \res_{B = 2}(a_{-2\a_1})$ equals
\begin{align*}
  \tfrac{288 m (m+1) (A-6+6 m) (A+6 m) (A-10+12 m) (A-6+12m)^2(A-2+12 m) (A+2+12 m) (A+6+12m)^2(A+10+12 m)}{(A-10)(A-6)^4(A-2)^2A(A+2)(A+6)}.
\end{align*}
Since $s_{\a_1}(-4\beta_2) = 4\beta_3$ and $s_{\a_1}(-2\a_1) = 2\a_1$, it follows by Lemma~\ref{lemma: symmetry of the residues} that the residue of $a_{4\beta_3} + a_{2\a_1}$ at $B=2$ is also zero.
Next we note that $s_{\beta_1}(2\beta_2) - 2\beta_1 = 2\a_2$, and we compute that $\res_{B = 2}(a_{2\beta_2}) = - \res_{B = 2}(a_{2\a_2})$ equals
\begin{align*}
  \tfrac{36 m (m+1)^2 (3m-1) (3m+1)  (A-26) (A+30) (A-10-12 m) (A-6-12 m) (A-2-12 m) (A+6-12m)^2 (A-6 m) (A+6+12 m)}{(A - 6)(A-2)^2A(A+2)(A+6)^4}.
\end{align*}
Since $s_{\a_1}(2\beta_2) = -2\beta_3$ and $s_{\a_1}(2\a_2) = -2\a_3$, it follows by  Lemma~\ref{lemma: symmetry of the residues} that the residue of $a_{ -2\beta_3} + a_{-2\a_3}$ at $B=2$ is also zero.
Thus, by Theorem~\ref{thm: sufficient conditions for ring preservation} part 1 there is no singularity at $B = 2$ in ${\mathcal D}_2p(z)$.  
 
Let us now consider possible singularities in ${\mathcal D}_2p(z)$ at $A=c \ge0$. By Theorem \ref{thm: sufficient conditions for ring preservation} part 1 and the form of the coefficients  \eqref{eq: a4beta}  -- \eqref{eq: a2alpha} it is sufficient to consider the case $A=6$ (equivalently, $\IP{\a_1}{z} = \a_1^2$).
In this case $\IP{\beta_2}{z} = \frac12 B + 3$, $\IP{\beta_3}{z} = -\frac12 B + 3$, $\IP{\alpha_2}{z} = -\frac32 B + 3$ and $\IP{\alpha_3}{z} = \frac32 B + 3$.
The only $\tau \in AG_2$ for which $a_\tau$ is singular at $A = 6$ and for which the corresponding $\lambda = s_{\a_1}(\tau) - 2\a_1 \neq 0$ are 
$\tau = -4\beta_2, -4\beta_3$ and $\pm2\beta_1$. 
We note that $s_{\a_1}(-4\beta_2) - 2\a_1 = -2\beta_1$, and we compute that 
$\res_{A = 6}(a_{-4\beta_2}) = - \res_{A = 6}(a_{-2\beta_1})$ equals 
\begin{align*}
    \tfrac{96 m(m+1) (B-14-12 m) (B-2-12 m) (B-2+4 m) (B+ 2+4 m) (B-2+6 m) (B+4+6 m) (B-6+12 m) (B+2+12 m) (B+6+12 m) (B+14+12 m)}{(B-6)^2(B - 2)^4 B (B+2)^2(B + 6)}.
\end{align*}
Since $s_{\beta_1}(-4\beta_2) = -4\beta_3$ and $s_{\beta_1}(-2\beta_1) = 2\beta_1$, it follows by Lemma~\ref{lemma: symmetry of the residues} that the residue of $a_{-4\beta_3} + a_{2\beta_1}$ at $A = 6$ is also zero.
By Theorem~\ref{thm: sufficient conditions for ring preservation} part 1 there is thus no singularity at $A = 6$ in ${\mathcal D}_2p(z)$.

By Corollary~\ref{cor: symmetry of the residues} it follows that 
${\mathcal D}_2p(z)$ has no singularities.
The proof that ${\mathcal D}_2p(z)$ belongs to $\mathcal{R}^a_{AG_2}$ can be completed in an analogous way to how it was done for the   operator~\eqref{eq: another difference operator for AG2} in the proof of Theorem \ref{thm: another operator for AG2 preserves the ring}.
\end{proof}

We now give a second construction of the \BA for $AG_2$.

\begin{theorem}\label{thm: construction of BA for AG2}
Let $M = \sum_{\gamma \in AG_{2, +}} m_\gamma =  12m+3.$
For $x \in \cplane$, let
\begin{equation*}
    \mu(x) =  \sum_{\tau \colon \frac12 \tau \in AG_2} \lambda_\tau(\exp\IP{x}{\tau}-1),
\end{equation*}
and
\begin{equation}\label{eq: c(x) for AG2}
    c(x) = \frac{M!}{8} \prod_{\gamma \in G_{2, +}} \left(\sum_{\tau \colon \frac12 \tau \in AG_2} \lambda_\tau \IP{\tau}{\gamma} \exp{\IP{x}{\tau}} \right)^{m_\gamma + m_{2\gamma}}.
\end{equation}
We also define for $z \in \cplane$ the polynomial belonging to $\mathcal{R}^a_{AG_2}$ given by
\begin{equation*}
    Q(z) = \prod_{\substack{\gamma \in G_{2, +} \\ s \in A_\gamma}}  \left(\IP{\gamma}{z}^2-s^2\gamma^2\right), 
\end{equation*}
where we recall $A_{\gamma} = \{1,2,3, \dots, m_\gamma, m_\gamma + 2m_{2\gamma}\}$ for all $\gamma \in G_{2,+}$.
Then the function
\begin{equation}\label{eq: formula for BA for AG2}
    \psi(z, x) = c^{-1}(x) ({\mathcal D}_2 - \mu(x))^M [Q(z) \exp\IP{z}{x}]
\end{equation}
is the BA function for $AG_2$.  
Moreover, $\psi$ is also an eigenfunction of the operator ${\mathcal D}_2$ with ${\mathcal D}_2 \psi = \mu(x) \psi$, thus bispectrality holds.
\end{theorem}

The proof  is similar to the proof of Theorem~\ref{thm: another construction of BA for AG2} and it can be found in \cite{MV}.

\begin{remark}
One can show that operators \eqref{eq: another difference operator for AG2} -- \eqref{eq: another a2alpha}, \eqref{eq: difference operator for AG2} -- \eqref{eq: a2alpha} commute: $[{\mathcal D}_1, {\mathcal D}_2]=0$. This can be proven by taking the rational limit of the more general trigonometric versions of these operators, which also commute \cite{FV2}.
\end{remark}

\section{Relation with $A_2$ and $A_1$ Macdonald--Ruijsenaars systems}
\label{Sect-A1-A2}

In the case when $m=0$, the configuration $AG_2$ reduces to the root system $\{\pm 2\beta_i\colon i=1,2,3\}$ of type $A_2$ with multiplicity $1$ for all vectors.
In this limit, the operator~\eqref{eq: difference operator for AG2} reduces to the quasiminuscule operator for (twice) this root system.
Let us now consider the $m=0$ limit of the operator~\eqref{eq: another difference operator for AG2}. After a rescaling, this gives an operator 
of the form 
\begin{equation}\label{eq: m=0 limit}
    D_0 = \sum_{\tau \in G_2} a_{\tau, 0}(z) (T_\tau - 1) = -24 + \sum_{\tau \in G_2} a_{\tau, 0}(z) T_\tau, 
\end{equation}
where for $\tau  = \varepsilon\beta_j$, $\varepsilon \in \{ \pm 1\}$, $j=1, 2, 3,$ we have
\begin{equation*}
    a_{\tau, 0}(z) = 
    3 
    \prod_{\substack{\gamma \in W \beta_1 \\ \IP{\tau}{(2\gamma)^\vee} = \frac12}} 
    \left( 1-\frac{\frac12 \gamma^2}{\IP{\gamma}{z} - \frac12\gamma^2} \right)
    \prod_{\substack{\gamma \in W \beta_1 \\ \IP{\tau}{(2\gamma)^\vee} = 1}}
    \left(1-\frac{\gamma^2}{\IP{\gamma}{z}}\right),
\end{equation*}
and for $\tau = \varepsilon\a_j$ we have
\begin{equation*}
    a_{\tau, 0}(z) = 
    \prod_{\substack{\gamma \in W \beta_1 \\ \IP{\tau}{(2\gamma)^\vee} = \frac32}} 
    \left( 1-\frac{\frac32 \gamma^2}{\IP{\gamma}{z} + \frac12\gamma^2} \right).
\end{equation*}

\begin{proposition}
The operator~\eqref{eq: m=0 limit} preserves the ring of analytic functions $p(z)$ satisfying $p(z+\beta_i)$ $=p(z-\beta_i)$ at $\IP{\beta_i}{z} = 0$ for all $i=1, 2, 3$.
\end{proposition}

The proof is parallel to the proof of Theorem \ref{thm: another operator for AG2 preserves the ring}.
In this case, though, the condition 2(b) of Theorem \ref{thm: sufficient conditions for ring preservation} is needed while it does not play a role in the proofs of Theorems \ref{thm: another operator for AG2 preserves the ring} and \ref{secondDpreserves}.

Let us rewrite the operator $D_0$ for
the more standard realisation of the root system $A_2$ given by $A_2 = \{ e_i - e_j\colon 1 \leq i \neq j \leq3 \} \subset \R^3$, where $e_i$ are the standard basis vectors.

\begin{proposition}
Define the set $S = S_1 \cup S_2$,  where 
$$S_1 = \{3e_i \colon i = 1, 2, 3 \} \cup \{2e_i + 2e_j - e_k: 1 \leq i < j \neq k \leq 3, i \neq k \} \textnormal{ and } S_2 = \{ 2e_i + e_j\colon 1 \leq i \neq j\leq 3\}.$$
Then the operator acting in the variable $z = (z_1, z_2, z_3)\in \C^3$ given by
\begin{equation}\label{eq: m=0 limit rewritten}
    \begin{aligned}
        \widetilde D_0 &= 3\sum_{\tau \in S_2} \left(\prod_{\substack{i \neq j \\ \IP{\tau}{e_i-e_j} = 1}}
        \left( 1-\frac{1}{z_i - z_j-1} \right) \prod_{\substack{i \neq j \\ \IP{\tau}{e_i-e_j} = 2}}
        \left( 1-\frac{2}{z_i - z_j} \right)\right) T_\tau \\
        & +\sum_{\tau \in S_1}\left( \prod_{\substack{i \neq j \\ \IP{\tau}{e_i - e_j} = 3}}
        \left( 1-\frac{3}{z_i - z_j+1} \right)\right) T_\tau.
    \end{aligned}
\end{equation}
preserves the ring of analytic functions $p(z)$ satisfying $p(z+e_i - e_j) = p(z-e_i + e_j)$ at $z_i=z_j$ for all $i, j =1, 2, 3$.
\end{proposition}
Note that $S_1 = \{\tau \in S\colon |\IP{\tau}{e_i - e_j}| \in \{ 0, 3\} \textnormal{ for all } i, j=1, 2, 3\}$ and 
$S_2 = \{\tau \in S\colon |\IP{\tau}{e_i - e_j}| \in \{0, 1, 2\}  \textnormal{ for all } i, j=1, 2, 3\}$.

Let us now consider a
 version of the operator~\eqref{eq: m=0 limit rewritten} for the root system $A_1$. 
Let $\sim$ denote equality of operators when acting on functions constant along the direction normal to the hyperplane $z_1 + z_2 = 0$.

\begin{proposition}
Let $S_1' = \{3e_1, 3e_2\}$ and $S_2' = \{2e_1 + e_2, e_1 + 2e_2\}.$ Then formula~\eqref{eq: m=0 limit rewritten} after replacement of $S_i$ with $S_i'$, $i=1, 2$, gives an operator $\widehat D_0$  acting   in the variable $z = (z_1, z_2) \in \C^2$ that
preserves the ring $\mathcal R^a_{A_1}$ of analytic functions $p(z)$ satisfying $p(z+e_1 - e_2) = p(z-e_1 + e_2)$ at $z_1=z_2$.
Moreover, if we split the operator  $\widehat D_0 = D_1 + D_2$, where 
\begin{equation*}
     D_1 = 3\left( 1 - \frac{1}{z_2 - z_1 - 1} \right) T_{e_1 + 2e_2} + \left(1-\frac{3}{z_1 - z_2 +1} \right) T_{3e_1}
\end{equation*}
and 
\begin{equation*}
     D_2 = 3 \left( 1 - \frac{1}{z_1 - z_2 - 1} \right) T_{2e_1 + e_2} + \left(1-\frac{3}{z_2 - z_1 + 1} \right) T_{3e_2},
\end{equation*}
then $D_i(\mathcal R^a_{A_1}) \subseteq \mathcal R^a_{A_1}$ for $i=1, 2$. The operators $D_i$ satisfy commutativity relations
$$
[D_1, D_2] = [D_1, D^{msl}] = [D_2, D^{msl}] = 0,
$$
where $D^{msl}$ is the operator for the minuscule weight $2 e_1$ of the root system $2 A_1$ with multiplicity 1 given by
\begin{equation*}
    D^{msl} = \left( 1 - \frac{2}{z_1 - z_2} \right) T_{2e_1} +  \left( 1 - \frac{2}{z_2 - z_1} \right) T_{2e_2}.
\end{equation*}
We also have
${\widehat D_0}^2 \sim (D^{msl} + 2)^3$, and $D_1 D_2 \sim 3 D^{qm} + 16 \sim 3(D^{msl})^2 + 4$, where
\begin{align*}
    D^{qm} & = \left( 1 - \frac{2}{z_1 - z_2} \right)\left( 1 - \frac{2}{z_1 - z_2+2} \right) (T_{4e_1}-1) +  
    \left( 1 - \frac{2}{z_2 - z_1} \right)\left( 1 - \frac{2}{z_2 - z_1+2} \right) (T_{4e_2} - 1)
\end{align*}
is the operator for the quasiminuscule weight $4e_1$.
\end{proposition}

We note that the operators $D_1$ and $D_2$ are not symmetric under the swap of the variables $z_1, z_2$. The operator $D^{msl}$ is symmetric and all three operators commute.

\vspace{10mm}
 {\bf Acknowledgments}. 
 M.V.\ acknowledges support received from the Institute of Mathematics and its Applications through their Small Grant Scheme, as well as support from the School of Mathematics and Statistics of the University of Glasgow for which we would like to thank Ian Strachan. The work of M.V.\ was also partially supported by a Carnegie PhD Scholarship from the Carnegie Trust for the Universities of Scotland.


\newpage
\begin{thebibliography}{99}
    \bibitem{BPF} Berezin, F. A.; Pokhil, G. P.; Finkelberg, V. M. ``Schr\"odinger equation for a system of one-dimensional particles with point interaction'', \textit{Vestnik Moskov. Univ. Ser. 1 Mat. Meh.} \textbf{1}, (1964), pp. 21--28.
    
	\bibitem{C'71} Calogero, F. ``Solution of the one-dimensional $n$-body problem with quadratic and/or inversely quadratic pair potential'', \textit{J.~Math.~Phys.} \textbf{12}, (1971) pp. 419--436.
	
	\bibitem{CV} Chalykh, O. A.; Veselov, A. P. ``Commutative rings of partial differential operators and Lie algebras", \textit{Comm. Math. Phys.}, vol. \textbf{126} (1990), no. 3, pp. 597--611.
	
	\bibitem{CSV} Chalykh, O.~A.; Styrkas, K.~L.; Veselov, A. P. ``Algebraic integrability for the Schr\"odinger equation and finite reflection groups'', \textit{Theoret.~and Math.~Phys.} \textbf{94}, no. 2, (1993) pp. 182--197.
	
	\bibitem{ChDarboux}
	Chalykh, O.~A.  ``Darboux transformations for multidimensional Schr\"odinger operators'',
    \textit{Russian Math. Surveys},  \textbf{53}, no.~2, (1998), pp.~167--168.

	
	\bibitem{CFV'98} Chalykh, O.; Feigin, M.; Veselov, A. ``New integrable generalizations of Calogero--Moser quantum problem'', \textit{J. Math. Phys.} \textbf{39}, no.~2, (1998) pp.~695--703. 
	
	\bibitem{CFV'99} Chalykh, O.~A.; Feigin, M.~V.; Veselov, A.~P. ``Multidimensional Baker-Akhiezer functions and Huygens' principle'', \textit{Comm. Math. Phys.} \textbf{206}, no. 3, (1999) pp. 533--566.


    \bibitem{Chbisp} Chalykh, O.~A. ``Bispectrality for the quantum Ruijsenaars model
    and its integrable deformation'', \textit{J. Math. Phys.} \textbf{41}, no. 8, (2000) pp.~5139--5167. 
    
    \bibitem{CVlocus} Chalykh, O.~A.; Veselov A.~P. ``Locus configuratons and $\vee$-systems'', \textit{Phys.~Let.~A.} \textbf{285}, no.~5-6, (2001) pp.~339--349.
    
    
    \bibitem{Ch07} Chalykh, O. ``Algebro-geometric Schr\"odinger operators in
    many dimensions'', \textit{Phil. Trans. R. Soc. A} \textbf{366}, (2008) pp. 947--971.
	
	
	\bibitem{Duijstermaat} Duistermaat, J.~J.; Gr\"unbaum, F.~A. ``Differential equations in the spectral parameter'', \textit{Comm.~Math.~Phys.} \textbf{103}, (1986) pp.~177--240.
	
	
	\bibitem{FF} Fairley, A.; Feigin, M. ``Trigonometric planar real locus configurations'', in preparation.
	
	\bibitem{Fbisp} Feigin M. ``Bispectrality for deformed
Calogero--Moser--Sutherland systems'', \textit{ J. Nonlin. Math. Phys.} \textbf{12}, suppl.~2, (2005) pp.~95--136. 


\bibitem{FVr} Feigin, M.; Vrabec, M. ``Intertwining operator for $AG_2$ \CMS system'',\\ \textit{J. Math. Phys.} \textbf{60}, (2019). 
	
\bibitem{FV2} Feigin, M.; Vrabec, M. ``Macdonald--Ruijsenaars $AG_2$ system'', in preparation. 


\bibitem{GVZ} Gorsky, A.; Vasilyev, M.; Zotov, A. ``Dualities in quantum integrable many-body systems and integrable probabilities -- I'', preprint, arXiv:2109.05562.

	

	
\bibitem{KK} Kharchev, S.; Khoroshkin, S. ``Wave function for $GL(n,\R)$ hyperbolic Sutherland model II. Dual Hamiltonians'', preprint, arXiv:2108.05393.

\bibitem{Kr} Krichever, I. M. ``Methods of algebraic geometry in the theory of non-linear equations'', Russian Math. Surveys, {\bf 32}, no. 6, (1977) pp. 185--213.





\bibitem{Macdonald2} Macdonald, I.~G. ``Orthogonal polynomials associated with root systems'', \textit{S\'emin.~Lothar.~Comb.} \textbf{70}, (2000) Article B45a.
	
\bibitem{M'75} Moser, J. ``Three integrable Hamiltonian systems connected with isospectral deformations'', \textit{Adv.~Math.} \textbf{16}, (1975) pp. 197--220.

\bibitem{OP'76} Olshanetsky, M.~A.; Perelomov, A.~M. ``Completely integrable Hamiltonian systems connected with semisimple Lie algebras'', \textit{Invent.~Math.} \textbf{37(2)}, (1976) pp. 93--108. 
	
\bibitem{OP'78} Olshanetsky, M.~A.; Perelomov, A.~M. ``Quantum systems related to root systems and radial parts of Laplace operators'', \textit{Func.~Anal.~Appl.} \textbf{12}, (1978) pp. 121--128.
	


\bibitem{RuijsSchn} Ruijsenaars, S.~N.~M.; Schneider, H. ``A new class of integrable systems and its relation to solitons'', \textit{Ann.~Phys.} \textbf{170}, (1986) pp.~370--405.

\bibitem{Ruijs} Ruijsenaars, S.~N.~M. ``Complete integrability of relativistic Calogero--Moser systems and elliptic function identities'', \textit{Commun.~Math.~Phys.} \textbf{110}, (1987) pp.~191--213.


\bibitem{Ruijs2} Ruijsenaars, S.~N.~M. ``Action-angle maps and scattering theory for some finite-dimensional integrable systems. I. The pure soliton case'', \textit{Comm.~Math.~Phys.} \textbf{115}(1), (1988) pp.~127--165. 

\bibitem{Ruijs3} Ruijsenaars, S.~N.~M. ``Finite-dimensional soliton systems'', in \textit{Integrable and Superintegrable systems}, edited by B.~Kupershmidt, World Scientific, Singapore, (1990) pp.~165--206.


\bibitem{SV1} Sergeev, A.~N.; Veselov, A.~P. ``Deformed quantum Calogero--Moser problems and Lie superalgebras'', \textit{Comm. Math. Phys.} \textbf{245}, no.~2, (2004) pp. 249–-278.

	
\bibitem{S'72} Sutherland, B. ``Exact results for a quantum many-body problem in one dimension, II'', \textit{Phys.~Rev.~A} \textbf{5}, (1972) pp. 1372--1376.


\bibitem{CFV'96} Veselov, A.~P.; Feigin, M.~V.; Chalykh, O.~A. ``New integrable deformations of the Calogero--Moser quantum problem'', \textit{Russ. Math. Surv.} \textbf{51}, (1996)  pp.~185--186.

\bibitem{MV} Vrabec, M.~V.; ``Generalised quantum Calogero--Moser--Sutherland systems'', Master Thesis, School of Mathematics and Statistics, University of Glasgow, (2020). 

\end{thebibliography}
\end{document}